\documentclass{elsarticle}

\let\today\relax
\makeatletter
\def\ps@pprintTitle{%
    \let\@oddhead\@empty
    \let\@evenhead\@empty
    \def\@oddfoot{\footnotesize\itshape
         {} \hfill\today}%
    \let\@evenfoot\@oddfoot
    }
\makeatother

\usepackage{amssymb}
\usepackage{amsmath}

\journal{}
\usepackage{natbib}

\usepackage{adjustbox}

\usepackage[skipabove=8pt,skipbelow=4pt,innertopmargin=6pt,innerbottommargin=6pt]{mdframed}
\usepackage[T1]{fontenc}
\usepackage[utf8]{inputenc}
\usepackage{scrextend}
\usepackage{hyperref} 
\usepackage{ragged2e}

\usepackage{mathtools}
\usepackage{amsmath}%,exscale,latexsym,float,eucal}
\usepackage{tikz, amssymb,bm,color, amsthm}

\usepackage{longtable}
\usepackage{makecell}
\usepackage[pagewise]{lineno}
\usetikzlibrary{arrows}
\usetikzlibrary{positioning, calc,chains,fit,shapes}
\usepackage{enumerate}
\usepackage{comment}
\usepackage{graphics}
\usepackage{longtable}
\usepackage{mdframed}
\usepackage{slashbox}
\usepackage{url}
\usepackage{framed}
\usepackage{array}
\usepackage{tabu}
\usepackage{lscape}
\usepackage{multirow}
\usepackage{ulem}
\usepackage{multicol}
\usepackage{placeins}
\usepackage{enumitem}

\usepackage{geometry}
\setcitestyle{notesep={: }}

\usepackage{natbib}

\mdfsetup{skipabove=2pt,skipbelow=2pt}
\newtheorem{theorem}{Theorem}
\newtheorem{proposition}[theorem]{Proposition}% 

\newtheorem{remark}{Remark}%

\newtheorem{definition}{Definition}%

\newtheorem{lemma}[theorem]{Lemma}
\newtheorem{observation}[theorem]{Observation}
\newtheorem{corollary}[theorem]{Corollary}

\newtheorem{construction}{Construction}

\newcommand{\pname}[1]{\textnormal{\textsc{#1}}}

\newcommand{\TSAT}{\pname{3-SAT}}

\newcounter{rowcntr}[table]
\renewcommand{\therowcntr}{\thetable.\arabic{rowcntr}}

\newcolumntype{N}{>{\refstepcounter{rowcntr}\therowcntr}c}

\AtBeginEnvironment{longtabu}{\setcounter{rowcntr}{0}}

\newcounter{rowcntra}[table]
\renewcommand{\therowcntra}{\arabic{rowcntra}}

\newcolumntype{M}{>{\refstepcounter{rowcntra}\therowcntra}c}

\AtBeginEnvironment{tabular}{\setcounter{rowcntra}{0}}

\newcommand{\XCD}{\pname{\raisebox{1.5pt}{$\chi$}$_{cd}$}}
\newcommand{\SC}{\pname{$\omega_{s}$}}
\newcommand{\NPC}{NP-complete}
\newcommand{\CDC}{\textsc{CD-coloring}}
\newcommand{\SCP}{\textsc{Separated-Cluster}}
\newcommand{\TD}{\textsc{Total Domination}}
\newcommand{\ISP}{\textsc{Independent Set}}
\newcommand{\CC}{\textsc{Clique Cover}}
\newcommand{\TETHS}{Further, the problem cannot be solved in time \ensuremath{2^{o(\mid V(G)\mid)}}, unless the ETH fails}
\usetikzlibrary{fit}
\usetikzlibrary{
  graphs,
  graphs.standard
}
\usepackage{subcaption}
 \usepackage[hang,flushmargin]{footmisc}

\PassOptionsToPackage{unicode}{hyperref}
\PassOptionsToPackage{hyphens}{url}

\usetikzlibrary {arrows.meta}

\usetikzlibrary{arrows,automata}

\usepackage{lmodern}
\usepackage{thmtools}

\usepackage{xcolor}
\usepackage{textcomp}
\begin{document}

\begin{frontmatter}


\title{Total Domination, Separated-Cluster, CD-Coloring: Algorithms and Hardness}

\renewcommand\rightmark{\textit{Total Domination, Separated-Cluster, CD-Coloring: Algorithms and Hardness}}
\renewcommand\leftmark{\textit{Total Domination, Separated-Cluster, CD-Coloring: Algorithms and Hardness}}

\author[1]{Dhanyamol Antony\corref{cor1}}\ead{dhanyamolantony@iisertvm.ac.in}

\author[2] {L. Sunil Chandran\corref{cor1}}\ead{sunil@iisc.ac.in}

\author[3]{Ankit Gayen\corref{cor1}}\ead{ankit.gayen@ens-lyon.fr}

\author[2]{Shirish Gosavi\corref{cor1}}\ead{shirishgp@iisc.ac.in}

\author[4]{Dalu Jacob\corref{cor1}}\ead{dalujacob@maths.iitd.ac.in}

\cortext[cor1]{Corresponding author}

\affiliation[1]{organization={School of Data Science, Indian Institute of Science Education and Research Thiruvananthapuram},%Department and Organization
            %addressline={}, 
           % city={},
            postcode={695551}, 
            state={Kerala},
            country={India}}
\affiliation[2]{organization={Computer Science and Automation department, Indian Institute of Science},%Department and Organization
            city={Bengaluru},
            postcode={560012}, 
            state={Karnataka},
            country={India}}
\affiliation[3]{organization={D\'{e}partment d'informatique, \'{E}cole normale sup\'{e}rieure de Lyon},%Department and Organization
         %   addressline={}, 
           % city={},
         %   postcode={}, 
       %     state={},
            country={France}}
\affiliation[4]{organization={Department of Mathematics, Indian Institute of Technology Delhi},%Department and Organization
            city={New Delhi},
            postcode={110016}, 
            country={India}}

%% Abstract
\begin{abstract}
% shirish Abstract text.
Domination and coloring are two classic problems in graph theory. In this paper, our major focus is on the  \CDC\ problem, which incorporates the flavors of both domination and coloring in it. Let $G$ be an undirected graph without isolated vertices.  A proper vertex coloring~$c:V(G)\rightarrow\{1,2,\ldots, k\}$ of $G$ is said to be a \textit{cd-coloring}, if for each color class $C_j$, where $j\in \{1,2,\ldots,k\}$, with respect to $c$, there exists a vertex $v_j$ in $G$ such that $C_j\subseteq N(v_j)$. The minimum integer $k$ for which there exists a \textit{cd-coloring} of $G$ using $k$ colors is called the  \textit{cd-chromatic number} of $G$, denoted as $\XCD(G)$. A set $S\subseteq V(G)$ is said to be a \textit{total dominating set}, if any vertex in $G$ has a neighbor in $S$. The \textit{total domination number} of $G$, denoted as \raisebox{1.5pt}{$\gamma$}$_t(G)$, is defined to be the minimum integer $k$ such that $G$ has a total dominating set of size $k$. A set $S\subseteq V(G)$ is said to be a \textit{separated-cluster} (also known as \textit{sub-clique}) if no two vertices in $S$ lie at a distance exactly 2 in $G$. The \textit{separated-cluster number} of $G$, denoted as $\SC(G)$, is defined to be the maximum integer $k$ such that $G$ has a separated-cluster of size $k$.

In this paper, we contribute to the literature connecting \CDC\ with \TD\ and \SCP. 
%\begin{addmargin}[2.5em]{2.5em}
\begin{enumerate}[label=(\roman*)]
    \item It is known that \TD\ is \NPC\ for \textit{cubic graphs and triangle-free subcubic graphs}. We strengthen this result by proving that both the problems {\sc CD-Coloring} and \TD\ are \NPC, and do not admit any subexponential-time algorithms on \textit{triangle-free $d$-regular graphs, for each fixed integer $d\geq 3$}, assuming the Exponential Time Hypothesis.
    
    \item For any graph $G$, it is easy to see that $\XCD(G)\geq \SC(G)$. Analogous to the well-known notion of \textit{`perfectness'}, we introduce the notion of \textit{`cd-perfectness'}. We prove a sufficient condition for a graph $G$ to be \textit{cd-perfect} (i.e. $\XCD(H)= \SC(H)$, for any induced subgraph $H$ of $G$). Our sufficient condition turns out to be necessary for certain graph classes (like \textit{triangle-free} graphs).  
    
    \item The notion of `$cd$-perfectness' provides a framework to study the algorithmic complexity of \CDC\ and \SCP. By using this framework we provide both positive and negative results concerning the algorithmic complexity of \CDC\ and \SCP.
    
    \item We also settle an open question by proving that the \SCP\ problem is solvable in polynomial time for the class of \textit{interval graphs}. 
\end{enumerate}
%\end{addmargin}

\end{abstract}

%% Keywords
\end{frontmatter}



\section{Introduction}
\label{sec:intro}
Graph theorists are always fascinated to study the relationship among correlated graph parameters as well as their structural features and explore their algorithmic consequences on the graphs for which these parameters coincide.
For instance, the \textit{chromatic number} \raisebox{1.5pt}{$\chi$} and the \textit{clique number} $\omega$. It is a well-known fact that for any graph $G$, \raisebox{1.5pt}{$\chi$}$(G)\geq \omega(G)$. 
Perfect graphs are the graphs $G$ having the property that for any induced subgraph $H$ of $G$, \raisebox{1.5pt}{$\chi$}$(H)= \omega(H)$. The notion of perfectness unifies the results concerning \textit{colorings} and \textit{cliques} for many important graph classes. The celebrated \textit{`Strong Perfect Graph Theorem'} ~\cite{ChudnovskyStrong06} gives a different perspective on perfect graphs by characterizing them by their structure instead of parameters.
Along these lines of research, here we explore the interconnections between a few correlated graph parameters. Domination and coloring are two important and well-motivated problems in graph theory. The central problem, `\textit{cd-coloring},' in this paper incorporates the flavors of both domination and coloring. Even though the other two problems studied in this paper, `\textit{total domination}' and \textit{`separated-cluster'}  have their own significance and are of independent interest, they share an interesting relationship with the `\textit{cd-coloring}' problem. In this paper, we explore these relationships in detail and obtain several exciting algorithmic consequences. 

 Let $G$ be an undirected graph without isolated vertices and $n=\mid V(G)\mid$. A \textit{proper vertex coloring} $c:V(G)\rightarrow \{1,2,\ldots,k\}$ of $G$ is the partitioning of the vertex set into $k$ color classes, say $C_1, C_2,\ldots, C_k$ such that for each $i\in \{1,2,\ldots,k\}$, $C_i$ is an independent set in $G$. Then $c$ is said to be a \textit{cd-coloring} of $G$ if, for each $j\in \{1,2,\ldots,k\}$, there exists a vertex $v_j\in V(G)$ such that $C_j\subseteq N(v_j)$. i.e., each class $C_j$ in $G$ has to be \textit{dominated} by a vertex $v_j\in V(G)\setminus C_j$. Hence the name \textit{class-domination} coloring, which we shortly call $cd$-coloring. 
 It is easy to see that if each vertex in $G$ is assigned a distinct color, then it is a $cd$-coloring of $G$ using $n$ colors. The minimum integer $k$ for which there exists a \textit{cd-coloring} of $G$ using $k$ colors is called the  \textit{cd-chromatic number} of $G$, denoted as $\XCD(G)$. Given an input graph $G$, the problem \CDC\ seeks to find the \textit{cd-chromatic number} of $G$.
\CDC\  is known to be \NPC\ for several special classes of graphs, including bipartite graphs~\cite{MerouaneHCK15} and chordal graphs~\cite{ShaluVSCDComplex20}. It is polynomial-time solvable for graph classes like trees~\cite{ShaluKiruCDTreeCobipar21}, co-bipartite graphs~\cite{ShaluKiruCDTreeCobipar21}, split graphs~\cite{MerouaneHCK15}, and claw-free graphs~\cite{ShaluVSCDComplex20}. Shalu et al.~\cite{ShaluVSCDComplex20} obtained a complexity dichotomy for \CDC\ for $H$-free graphs. \CDC\  is also studied in the paradigm of parameterized complexity~\cite{BanKasRamDomClusGr23,KritRaiST21}, and approximation complexity~\cite{ChenTheDC14}. In addition to its theoretical significance, \CDC\ has a wide range of practical applications in social networks~\cite{ChenTheDC14} and genetic networks~\cite{KlavTavaDomHerProd21}. 

A set $S\subseteq V(G)$ is said to be a \textit{total dominating set} if any vertex in $G$ has a neighbor in $S$. Note that the induced subgraph $G[S]$ does not contain an isolated vertex.
The \textit{total domination number} of $G$, denoted as \raisebox{1.5pt}{$\gamma$}$_t(G)$, is defined to be the minimum integer $k$ such that $G$ has a total dominating set of size $k$. Given an input graph $G$, in the problem \TD, we intend to find the total domination number of $G$. Total domination is one of the most popular variants of domination, and there are hundreds of research papers dedicated to this notion in the literature~\cite{MerouaneHCK15,HoppManTotDom19,HennYeoTotDom13,ZhuApproxMinTotDom09}.

A set $S\subseteq V(G)$ is said to be a \textit{separated-cluster} if no two vertices in $S$ lie at a distance of exactly 2 in $G$. The authors of~\cite{ShaluKiruCDTreeCobipar21,ShaluSandhyaLowBound17} call it a \textit{subclique}.   Since we have decided to call it a separated cluster instead of subclique, we would like to explain the reason for this deviation. We say that $S\subseteq V(G)$ induces a \textit{cluster} in $G$, if $G[S]$ is a disjoint union of cliques.
Now, if $S\subseteq V(G)$ has no two vertices lying at a distance exactly 2 in $G$, then we can see that $S$ would induce a cluster in $G$ with the following additional property; no two vertices belonging to two distinct cliques in the cluster $S$ have a common neighbor in $G$. Using this perspective, we can infer that \SCP\ is a variant of the well-known problem, {\sc Cluster Vertex Deletion} where we intend to find the minimum number of vertices whose deletion results in a disjoint union of cliques~\cite{ShamSharCluGrMod04} (or, alternatively, to find the maximum number of vertices that can be partitioned into disjoint union of cliques). The \textit{separated-cluster number} of $G$, denoted as $\SC(G)$, is defined to be the maximum integer $k$ such that $G$ has a separated-cluster of size $k$. Given an input graph $G$, in the problem \SCP,\ our goal is to find the \textit{separated-cluster number} of $G$. \SCP\ is known to be \NPC\ for graph classes like bipartite graphs, chordal graphs, $3K_1$-free graphs~\cite{ShaluSandhyaLowBound17}, and polynomial-time solvable for trees, co-bipartite graphs, cographs, split graphs~\cite{ShaluSandhyaLowBound17} etc.

\vspace{-0.25cm}
\subsection{\textbf{Total domination and $cd$-coloring}}
\vspace{-0.1cm}
Graph coloring variants that also incorporate the concepts of domination are well studied~\cite{ShaluKiruCDTreeCobipar21, ShaluSandhyaLowBound17, ChellaliVolkDomChr04,ChellaliMaffDom12}. A close connection between \textsc{CD-Coloring} and \textsc{Total Domination} was established by Merouane et al.~\cite{MerouaneHCK15}.  In particular, it is known that for any triangle-free graph $G$, \raisebox{1.5pt}{$\gamma$}$_t(G)=\XCD(G)$.  As a consequence, the complexity results of \TD\ on the subclasses of triangle-free graphs are also applicable to \CDC. For instance, the result that \TD\ is \NPC\ on bipartite graphs with bounded degree 3~\cite{ZhuApproxMinTotDom09} implies that \CDC\ is also \NPC\ on this graph class. On the other hand, even though  \TD\ is known to be \NPC\ on cubic graphs~\cite{GarJohIntract79}, the algorithmic complexity of \TD\ on triangle-free cubic graphs, in fact, triangle-free $d$-regular graphs for each $d\geq 1$ is not known. Since regular graphs may not be triangle-free in general, the algorithmic complexity of \CDC\  for cubic graphs can not be inferred from the results on \TD.    

In this paper, we study the \CDC\ on regular graphs and prove the following theorem.

\begin{theorem}
    \label{thm:regular}
    \CDC\ is \NPC\ on triangle-free $d$-regular graphs, for each fixed integer $d\geq 3$. \TETHS.
\end{theorem}
   It is fascinating to observe how \textit{cd-coloring},  a relatively new graph problem, sheds light on a classic problem of total domination. For instance, as an additional advantage, we can use Theorem~\ref{thm:regular} to obtain the following corollary which gives results on \TD\ by recalling the fact that for any triangle-free graph $G$, \raisebox{1.5pt}{$\gamma$}$_t(G)=\XCD(G)$. It is remarkable that the following corollary is stronger than the existing results of \TD\ on regular graphs.
 
 \begin{corollary}
\label{thm:total_domination_regular}
\TD\ on triangle-free $d$-regular graphs is \NPC, for any constant $d\geq 3$. \TETHS.    
\end{corollary}

%\vspace{-0.5cm}
\subsection{\textbf{Separated-Cluster and $cd$-coloring}}
\vspace{-0.05cm}
Interestingly, we can see that the  \textit{cd-chromatic number} $\XCD$ and \textit{separated-cluster number} $\SC$  follow a similar relationship as their classic counterparts of \textit{chromatic number} \raisebox{1.5pt}{$\chi$} and \textit{clique number} $\omega$. Since any two vertices in a separated-cluster of a graph $G$ cannot lie at a distance exactly 2 in $G$, we have $\XCD(G)\geq \SC(G)$ (as each vertex in $\SC(G)$ has to be in different color class). It is not difficult to see that, using the same idea of \textit{`Mycielskian'} construction~\cite{MycielColoriage55} (in classic coloring), we can have a family of graphs for which $\omega_s=2$ but $\XCD$ is arbitrarily large.

The following auxiliary graph construction is very useful in the study of separated-clusters and $cd$-coloring.

%\vspace{-0.5cm}
\begin{definition}[Auxiliary graph $G^*$~\cite{ShaluSandhyaLowBound17}]\label{def:aux}
    Given a graph $G$, the auxiliary graph $G^*$ is the graph having $V(G^*)=V(G)$ and $E(G^*)= \{uv:u,v\in V(G)$ and $d_G(u,v)=~2\}$, where  $d_G(u,v)$ denote the distance between the vertices $u$ and $v$ in $G$.
\end{definition}
%\vspace{-0.05cm}
%\vspace{-0.5cm}
From the definitions of the auxiliary graph $G^*$ of a graph $G$ and independence number  $\alpha(G^*)$  (the size of a maximum cardinality independent set in $G^*$), it is straightforward to observe that $\SC(G)=\alpha(G^*)$. Further, by the definition of the clique cover number $ k(G^*)$  (the minimum number of cliques needed to partition the vertex set of $G^*$), it is not difficult to infer that $\XCD(G)\geq k(G^*)$ (see Observation~\ref{obs:cliquecover} for details).
Note that the parameters $\alpha$ and $k$ are well-studied in the literature, particularly in connection with \textit{perfect graphs}.  \textit{A graph $G$ is perfect if, for any induced subgraph $H$ of $G$,  $\alpha(H)=k(H)$}.\\

\noindent\textbf{Our Results:} Motivated by the notion of perfect graphs in classic coloring theory, we introduce $cd$-perfectness.  
We say that a graph $G$ is \textit{$cd$-perfect} if for any induced subgraph $H$ of $G$, $\XCD(H)= \SC(H)$. Though some earlier researchers~\cite{ShaluKiruCDTreeCobipar21, ShaluSandhyaLowBound17} observed that each graph $G$ belonging to certain graph classes like co-bipartite graphs, trees etc. satisfy the condition that $\XCD(G)= \SC(G)$, it is in this paper, the notion of $cd$-perfectness is introduced for the first time. This general notion of \textit{cd-perfectness} allows us to unify several (existing) results concerning \textit{$cd$-coloring} and \textit{separated-clusters} for various graph classes.

As a first step to prove the structure of $cd$-perfect graphs, we prove a sufficient condition for a graph $G$ to satisfy the equality, $\XCD(G)= k(G^*)$  (see Theorem~\ref{thm:suff_cdclique}). 
This leads us to Theorem~\ref{thm:suff_cdsubclique}, where the collection of graphs $\mathcal{H}$ is as shown in Figure~\ref{fig:c6-free}. 

 \begin{figure}[!htbp]
 \centering
    \begin{subfigure}[b]{0.2\textwidth}
          \centering
          {\fbox{\begin{tikzpicture}[myv/.style={circle, draw, inner sep=1.5pt,line width=0.3mm}]
  \node (z) at (0,0) {};

  \node[myv] (a) at (-0.4,0.4) {};
  \node[myv] (b) at (0.4,0.4) {};
  \node[myv] (c) at (-0.4,0) {};
  \node[myv] (d) at (0.4,0) {};
  \node[myv] (e) at (-0.4,-0.4) {};
  \node[myv] (f) at (0.4,-0.4) {};

  \draw [line width=0.3mm]  (a) -- (b);
  \draw [line width=0.3mm] (b) -- (c);
  \draw [line width=0.3mm] (c) -- (d);
  \draw [line width=0.3mm] (d) -- (e);
  \draw [line width=0.3mm] (e) -- (f);
  \draw [line width=0.3mm] (f) -- (a);
\end{tikzpicture}}}  
          \caption{$C_6$}
          \label{fig:c6}
     \end{subfigure}
      \begin{subfigure}[b]{0.2\textwidth}
          \centering
          {\fbox{\begin{tikzpicture}[myv/.style={circle, draw, inner sep=1.5pt,line width=0.3mm}]
  \node (z) at (0,0) {};

  \node[myv] (a) at (-0.4,0.4) {};
  \node[myv] (b) at (0.4,0.4) {};
  \node[myv] (c) at (-0.4,0) {};
  \node[myv] (d) at (0.4,0) {};
  \node[myv] (e) at (-0.4,-0.4) {};
  \node[myv] (f) at (0.4,-0.4) {};

  \draw [line width=0.3mm](a) -- (b);
  \draw [line width=0.3mm](b) -- (c);
  \draw [line width=0.3mm](c) -- (d);
  \draw [line width=0.3mm](d) -- (e);
  \draw [line width=0.3mm](e) -- (f);
  \draw [line width=0.3mm](f) -- (a);

  \draw [line width=0.3mm] (b) -- (d);
\end{tikzpicture}}}  
          \caption{$C_6^1$}
          \label{fig:c61}
     \end{subfigure}
      \begin{subfigure}[b]{0.2\textwidth}
          \centering
          {\fbox{\begin{tikzpicture}[myv/.style={circle, draw, inner sep=1.5pt,line width=0.3mm}]
  \node (z) at (0,0) {};

  \node[myv] (a) at (-0.4,0.4) {};
  \node[myv] (b) at (0.4,0.4) {};
  \node[myv] (c) at (-0.4,0) {};
  \node[myv] (d) at (0.4,0) {};
  \node[myv] (e) at (-0.4,-0.4) {};
  \node[myv] (f) at (0.4,-0.4) {};

  \draw [line width=0.3mm] (a) -- (b);
  \draw [line width=0.3mm] (b) -- (c);
  \draw [line width=0.3mm] (c) -- (d);
  \draw [line width=0.3mm] (d) -- (e);
  \draw [line width=0.3mm] (e) -- (f);
  \draw [line width=0.3mm] (f) -- (a);

  \draw [line width=0.3mm] (b) -- (d);
  \draw [line width=0.3mm] (f) -- (d);
\end{tikzpicture}}}  
          \caption{$C_6^2$}
          \label{fig:c62}
     \end{subfigure}
      \begin{subfigure}[b]{0.2\textwidth}
          \centering
          {\fbox{\begin{tikzpicture}[myv/.style={circle, draw, inner sep=1.5pt,line width=0.3mm}]
  \node (z) at (0,0) {};

  \node[myv] (a) at (-0.4,0.4) {};
  \node[myv] (b) at (0.4,0.4) {};
  \node[myv] (c) at (-0.4,0) {};
  \node[myv] (d) at (0.4,0) {};
  \node[myv] (e) at (-0.4,-0.4) {};
  \node[myv] (f) at (0.4,-0.4) {};

  \draw [line width=0.3mm] (a) -- (b);
  \draw [line width=0.3mm] (b) -- (c);
  \draw [line width=0.3mm] (c) -- (d);
  \draw [line width=0.3mm] (d) -- (e);
  \draw [line width=0.3mm] (e) -- (f);
  \draw [line width=0.3mm] (f) -- (a);

  \draw [line width=0.3mm] (b) -- (d);
  \draw [line width=0.3mm] (f) -- (d);
  \draw [line width=0.3mm] (b) to [bend left=50] (f);
\end{tikzpicture}}}  
          \caption{$C_6^3$}
          \label{fig:c63}
     \end{subfigure}
     \caption{ Set of graphs in $\mathcal{H}$ }
    \label{fig:c6-free}
\end{figure}

\begin{theorem} \label{thm:suff_cdsubclique}
     Let $G$ be an $\mathcal{H}$-free graph and $G^*$ its  auxiliary graph. If $k(G^*)=\alpha(G^*)$ then $\XCD(G)=\SC(G)$. Consequently, if $G$ is $\mathcal{H}$-free and $G^*$ is perfect, then $\XCD(G)=\SC(G)$.
 \end{theorem}

As an immediate corollary of Theorem~\ref{thm:suff_cdsubclique}, we obtain a sufficient condition for a graph to be $cd$-perfect (see Corollary~\ref{corr:suffcdperfect}). 
 We then use this result to bring several graph classes, like \textit{co-bipartite graphs}, \textit{chordal bipartite graphs}, etc., under the common umbrella of \textit{$cd$-perfect} graphs. Even though these results are structural in nature, they have several exciting algorithmic consequences. It is interesting to note that the same framework can be used as a tool to derive both positive and negative results concerning the algorithmic complexity of \CDC\ and \SCP. In particular, we have the following results.

\begin{enumerate}[label=(\roman*)]
   
    \item A beautiful interplay between the problems, {\sc Total Domination}, \CDC, and \SCP\ can be witnessed in Theorem~\ref{thm:chordalbip}, where we use Corollary~\ref{corr:suffcdperfect}, Theorem~\ref{thm:totdom}, and the fact that \raisebox{1.5pt}{$\gamma$}$_t(G)=\XCD(G)$ for triangle free-graphs $G$, to prove that the three problems mentioned above are equivalent and solvable in $O(n^2)$ time for chordal bipartite graphs. Consequently, this result improves and generalizes the existing $O(n^3)$-time algorithm for \SCP\ on the class of $P_6$-free chordal bipartite graphs~\cite{ShaluKiruP5Free22}. 
%\vspace{-0.5cm}
\begin{theorem}\label{thm:chordalbip}
    Let $G=(A,B,E)$ be a chordal bipartite graph. Then $G$ is $cd$-perfect. Consequently, {\sc Total Domination}, \CDC, and \SCP\  are all equivalent problems for chordal bipartite graphs and can be solved in $O(n^2)$ time.
\end{theorem}
 %  \vspace{-0.5cm}

     \item  We provide a unified approach (most of the time with an improvement) for finding polynomial-time algorithms for \CDC\  on certain graph classes. For instance, we prove the following theorem. 
%\vspace{-0.5cm}
    \begin{theorem} \label{thm:properalpha}
The \CDC\ problem can be solved in $O(n^{2.5})$ time for proper interval graphs and $3K_1$-free graphs.
\end{theorem}
     
    \item We also derive the following hardness results for $C_6$-free bipartite graphs.
%\vspace{-0.5cm}
    \begin{theorem}\label{thm:C_6free hard}
    The problems  \CDC\ and \SCP\ are NP-Complete for $C_6$-free bipartite graphs.
\end{theorem}

 % \vspace{-0.5cm} 
    \item Further, as an additional result, we prove the following theorem concerning separated-cluster on interval graphs which settles an open problem in~\cite{ShaluSandhyaLowBound17}.
%\vspace{-0.5cm}
    \begin{theorem}
    \label{thm:sepiterval}
    The \SCP\ problem in interval graphs can be solved in polynomial time.
\end{theorem}

\end{enumerate}
\section{Preliminaries and Notations}
\label{sec:prelim}
All graphs considered in this paper are undirected, simple, and finite.  
Let $G$ be a graph. The vertex set and edge set of $G$ are denoted by $V(G)$ and $E(G)$, respectively.  A \textit{subgraph} $G'$ of $G$ is a graph such that $V(G')\subseteq V(G)$ and $E(G')\subseteq E(G)$. Let $S\subseteq V(G)$, then \textit{induced subgraph}, denoted as $G[S]$, is a graph whose vertex set is $S$ and whose edge set contains all the edges in $E(G)$ that have both their endpoints in $S$. Sometimes, we also call $G[S]$ to be the \textit{graph induced by the vertices} in $S$. By $\bar{G}$, we denote the \textit{complement} of a graph $G$, which is the graph having the same set of vertices as $G$ such that two vertices $u,v$ in $\bar{G}$ are adjacent if and only if they are non-adjacent in $G$. 
For a vertex $v\in G$, we define the \textit{neighborhood} of $v$ in $G$, denoted as $N_G(v)$, to be the set of all the vertices adjacent to $v$ in $G$ (sometimes, we omit the subscript, if the graph $G$ is clear from the context). The degree of a vertex $v$ in $G$, denoted by $d_G(v)$, is the number of edges incident on $v$ in $G$. The maximum degree of a graph $G$ is denoted by $\Delta(G)$.
We denote by $d_G(u,v)$, the \textit{distance} between two vertices $u$ and $v$ in $G$, which is the number of edges in the shortest path between $u$ and $v$ in $G$. The \textit{square} of a graph $G$, denoted by $G^2$, is the graph having $V(G^2)=V(G)$, and $E(G^2)=\{uv:d_G(u,v)\leq 2\}$.
 The \textit{disjoint union} of two graphs $G_1$ and $G_2$, denoted by $G_1\uplus G_2$, is the graph having
    $V(G_1\uplus G_2) = V(G_1)\uplus V(G_2)$ and $E(G_1\uplus G_2) = E(G_1) \uplus E(G_2)$. The disjoint union of $t$ copies of a graph $G$ is denoted by $tG$. Let $H$ be a subgraph of $G$. We denote by $G-H$, the graph obtained from $G$ by removing the vertices in $H$, i.e. $G-H = G[V(G)\setminus V(H)]$. 

A set $M\subseteq E(G)$ is said to be a \textit{matching} in $G$ if no two edges in $M$ share a common vertex. A set $S\subseteq V(G)$ is said to be a \textit{clique} (respectively an \textit{independent set}) if every pair of vertices in $S$ are adjacent (respectively non-adjacent) in $G$. 
The number of vertices in the largest clique (respectively independent set) of a graph $G$ is called the \textit{clique number} (respectively \textit{independence number}) of $G$, denoted as $\omega(G)$ (respectively $\alpha(G)$). 
A \textit{clique cover} of a graph $G$ is the partition of the vertices of $G$ into cliques. A clique $K$ in $G$ is said to be \textit{maximal} if $K\cup \{v\}$ does not form a clique for any vertex $v\in V(G)$. 
The \textit{clique cover number}, denoted by $k(G)$, is the smallest $k$ for which $G$ has a clique cover of size $k$.

A graph $G$ is said to be \textit{$H$-free}, if it does not contain $H$ as an induced subgraph. A \textit{diamond} is a graph on 4 vertices and 5 edges. i.e. it is the graph obtained by deleting an edge from $K_4$ (complete graph on 4 vertices). A graph $G$ is said to be \textit{d-regular} if every vertex of $G$ has degree $d$. A \textit{bipartite} graph, denoted as $G=(A,B,E)$, is a graph whose vertices can be partitioned into two independent sets $A$ and $B$ such that every edge contains its one end-point in $A$ and the other end-point in $B$. A \textit{complete bipartite} graph, denoted as $K_{m,n}$, is a bipartite graph $G=(A,B,E)$ with $\mid A\mid=m$, $\mid B\mid=n$ such that each vertex in $A$ is adjacent to every vertex in $B$. A \textit{star graph} is the complete bipartite graph, $K_{1,t}$ on $t+1$ vertices for some integer $t\geq 1$. A \textit{claw} $K_{1,3}$ is a star on 4 vertices. The complement of a bipartite graph is called \textit{co-bipartite} graph. A \textit{chordal bipartite graph} is a bipartite graph that does not contain any induced cycle of length $k$, where $k\geq 6$. A collection $\{I_v\}_{v\in V(G)}$, of intervals on a real line is said to be an \textit{interval representation} of a graph $G$ if for any pair of vertices $u,v\in V(G)$, we have $uv\in E(G)$ if and only if $I_u\cap I_v\neq \emptyset$. A graph is said to be an \textit{interval graph} if it has a corresponding interval representation. A \textit{proper interval} graph is an interval graph having an interval representation in which no interval properly contains another. 

Even though the parameters $\XCD$ and $\SC$ have a strong similarity to their classic counterparts, $\chi$, and $\omega$, here we note an important property of $\XCD$ and $\SC$, which is fundamentally different from $\chi$ and $\omega$. 

\noindent\textbf{Note:}
For any induced subgraph $H$ of a graph $G$, we have \raisebox{2pt}{$\chi$}$(H)$$\leq$ \raisebox{2pt}{$\chi$}$(G)$ and $\omega(H)\leq \omega(G)$. But it can be seen that this is \textit{not true} in general for the parameters $\XCD$ and $\SC$. i.e. it is possible for  $G$ to have an induced subgraph $H$ such that $\XCD(H)>\XCD(G)$ or $\SC(H)>\SC(G)$. For instance, let $G=K_{1,3}$, the star graph, with central vertex, say, $u$, and leaf vertices, say, $v_1,v_2$, and $v_3$. It is easy to verify that $\XCD(G)=\SC(G)=2$. On the other hand, consider the subgraph induced by leaf vertices, say, $H=G[\{v_1,v_2,v_3\}]$. We then have $\XCD(H)=\SC(H)=3>2=\SC(G)=\XCD(G)$.

In the following observation, we have a natural sufficient condition for an induced subgraph, $H$ of a graph $G$ to have \XCD($G$) $\geq$ \XCD($H$). We use this observation later.

\begin{observation}
    \label{obs:cdcolor_reduce}
   
      Let $H$ be an induced subgraph of the graph $G$. Suppose there exist no independent set  $I\subseteq V(H)$ of vertices such that $I$ has a dominating vertex in $G$ but not in $H$, then  \XCD($G$) $\geq$ \XCD($H$).
 
\end{observation}

\begin{proof} 
   Let \XCD($G$) =$k$, and $\{I_1,I_2,\ldots, I_k\}$ be the color classes in $G$ with respect to a $k$-$cd$-coloring of $G$.  Then the independent sets $\{I_{h1},I_{h2},\ldots, I_{hk}\}$ is a partition of $V(H)$, where $I_{hi}= I_{i}\cap V(H)$. Note that each of the sets in $\{I_1, I_2,\ldots, I_k\}$ is dominated by a vertex in $G$. From the statement of the observation, it is then clear that each of the sets in $\{I_{h1},I_{h2},\ldots, I_{hk}\}$ is also dominated by a vertex in $H$. This implies that $\XCD(H)\leq k=\XCD(G)$. 
\end{proof}

We also make use of the following known results.

\begin{proposition}[\cite{ZhuApproxMinTotDom09}]
    \label{pro:bipartite}
   Total domination in bipartite graphs with bounded degree 3 is  \NPC. \TETHS.
\end{proposition}

\begin{proposition}[\cite{MerouaneHCK15}]
    \label{pro:triangle-free}
   Let $G$ be a triangle-free graph. Then \XCD ($G$) = $\gamma_t (G)$.
\end{proposition}

The major problems that we consider in this paper are: \\

\begin{mdframed}
  \textbf{\sc{\TD\ } }\\ \textbf{Input:} A graph $G$ and an integer $k$\\ \textbf{Question:} Does there exist a \textit{total dominating }set $S\subseteq V(G)$ of size at most $k$ in $G$?
\end{mdframed}

\begin{mdframed}
  \textbf{\sc{\CDC\ } }\\ \textbf{Input:} A graph $G$ and an integer $k$\\ \textbf{Question:} Does there exist a $cd$-coloring of $G$ with at most $k$ colors?
\end{mdframed}

\begin{mdframed}
  \textbf{\sc{\SCP\ } }\\
  \textbf{Input:} A graph $G$ and an integer $k$\\ \textbf{Question:} Does there exist a \textit{seprarated-cluster} of size at least $k$ in $G$?  
\end{mdframed}

Further, we also consider the following problems for reduction purposes.
For an input graph $G$ and an integer $k$, the problem \CC\ (respectively \ISP) asks does there exists a clique cover (respectively an independent set) of size $k$ in $G$.
\section{Total domination and $cd$-coloring in triangle-free $d$-regular Graphs}
\label{sec:hardness}
In this section, we prove Theorem~\ref{thm:regular}, i.e. we derive the hardness result for \CDC\ in triangle-free $d$-regular graphs for each fixed integer $d\geq 3$. As a corollary of this result, we obtain the hardness result for \TD\ in triangle-free $d$-regular graphs for each fixed integer $d\geq 3$. 
Hence, first, we propose a simple linear reduction for \CDC\ on triangle-free cubic graphs. Then, we generalize this reduction for triangle-free $d$-regular graphs for each fixed integer $d\geq 3$, by using two linear reductions: one for odd values of $d$ and the other for even values.

The Exponential-Time Hypothesis (ETH) along with the Sparsification   Lemma implies that \TSAT\ cannot be solved in subexponential-time, i.e.
in time $2^{o(n+m)}$, where $n$ is 
the number of variables and $m$ is the number of clauses in the input formula.
 To prove that a problem does not admit a subexponential-time algorithm, %(algorithm runs in $2^{o(n)}$ algorithm, where $n$ is the size of input instance), 
 it is sufficient to obtain a linear reduction from a problem 
which does not admit a subexponential time algorithm where a linear reduction is a polynomial-time reduction
in which the size of the resultant instance is linear in the size of the input instance.
We refer to Chapter 14 of the book~\cite{CyganOthParamAlgo} for more details about these concepts.
All reductions mentioned in this section are trivially linear.

To prove Theorem~\ref{thm:regular}, we use three constructions. 
Construction~\ref{cons:cubic} is used for a reduction from \textsc{Total Domination}  on bipartite graphs with bounded degree 3 to \CDC\ on triangle-free cubic graphs. Construction~\ref{cons:d-regular odd} is used for a reduction from \CDC\ on triangle-free $(d-2)$-regular graphs to \CDC\ on triangle-free $d$-regular graphs, for each odd integer $d\geq 5$, whereas the Construction~\ref{cons:d-regular even} is used for a reduction from \CDC\ on triangle-free $(d-1)$-regular graphs to \CDC\ on triangle-free $d$-regular graphs, for each even integer $d\geq 4$. 
\begin{figure}[!htbp]
  \centering
    \centering
    \resizebox{0.45\textwidth}{!}{\begin{tikzpicture}[myv/.style={circle, draw, inner sep=3.5pt},myv1/.style={circle, draw, inner sep=0pt,color=white},myv2/.style={rectangle, draw,dotted,inner sep=0pt,line width = 0.5mm}]
 % \node [label=above: $\Large{r}$](z)  at (0,4) {};

  \node[myv] (a) at (0,3) {$a$};
  \node[myv][fill=black] (b) at (-2,2) {\textcolor{white}{$b$}};
  \node[myv][fill=black] (c) at (2,2) {\textcolor{white}{$c$}};

  \node[myv][fill=black] (d) at (-3,1) {\textcolor{white}{$d$}};
  \node[myv] (e) at (-1,1) {$e$};
  \node[circle, draw, inner sep=2.5pt][fill=black] (f) at (1,1) {\textcolor{white}{$f$}};
  \node[myv] (g) at (3,1) {$g$};

  \node[myv] (h) at (-3,-0.25) {$h$};
  \node[myv] (i) at (-1,-0.25) {$i$};
  \node[myv] (j) at (1,-0.25) {$j$};
  \node[myv] (k) at (3,-0.25) {$k$};
   \node[myv1] (l)  at (3,-1) {};
   \node[myv1] (m)  at (3,-1.5) {};
  
    \draw (a) -- (b);
    \draw (a) -- (c);
    \draw (b) -- (d);
    \draw (b) -- (e);
    \draw (c) -- (f);
    \draw (c) -- (g);
    \draw (d) -- (h);
    \draw (d) -- (i);
    \draw (e) -- (h);
    \draw (e) -- (i);
    \draw (f) -- (j);
    \draw (f) -- (k);
    \draw (g) -- (j);
    \draw (g) -- (k);
    \draw (i) to [bend right=20] (j);
    \draw (h) to [bend right=20] (k);

    \node[myv2][fit=(b) (d) (g) (h) (k) (m),  inner xsep=5.5ex, inner ysep=1.00ex, label=right:$H$] {}; 
    \node[myv2][dashed][fit= (d) (g) (h) (k) (l),  inner xsep=1.5ex, inner ysep=0.75ex, label=right:$H'$] {};
 % \draw (a) -- (d);
  
\end{tikzpicture}}
    \caption{The gadget $W$ used in Construction~\ref{cons:cubic}. The shaded vertices, $d,b,f$ and $c$, respectively dominates the color classes, $\{b,h,i\},\{e,a,d\},\{c,j,k\}$, and $\{g,f\}$, in a valid $cd$-coloring of $W$.}
    \label{fig:cubic_gadget}
  \end{figure}
  
\begin{construction}
\label{cons:cubic}
    Let $G$ be a bipartite graph with bounded degree 3. We construct a graph $G_c$ from $G$ %using a gadget $W$ 
    in the following way: 
    %\vspace{-0.25cm}
    \begin{itemize}
        \item First we construct a gadget $W$ (refer Figure~\ref{fig:cubic_gadget}) as follows: introduce a star graph $K_{1,2}$, which we denote by $S$, having $a$ as the root vertex, and $b,c$ as the leaves. Note that we also call `$a$' as the root vertex of $W$.
        Then we introduce two copies of the complete bipartite graph $K_{2,2}$, denoted by $C_1=(A_1,B_1)$ and $C_2=(A_2,B_2)$, respectively. Now the leaf $b$ (respectively $c$) of $S$ is adjacent to each  vertex of the partition $A_1$ (respectively $A_2$) of $C_1$ (respectively $C_2$).
    Furthermore,  each vertex of partition $B_1$ of $C_1$ is adjacent to exactly one vertex of the partition $B_{2}$ of $C_{2}$  as shown in Figure~\ref{fig:cubic_gadget}.  Let $H=W-\{a\}$ and $H'=H-\{b,c\}$.

    \item Thereafter we construct a graph $G_c$ from $G$ using the gadget $W$ as follows: for every vertex $u$ of $G$ with degree 2, introduce a gadget $W$ such that the root vertex $a$ of $W$ is adjacent to $u$ in $G_c$. Similarly, for each vertex $v$ in $G$ of degree 1, we introduce two copies of  $W$ such that the root vertices of both the gadgets are adjacent to $v$ in $G_c$.   \end{itemize}
\vspace{-0.15cm}    The graph $G_c- G$ contains $11x+22y$ vertices, where $x$ and $y$ denote the number of vertices in $G$ having degree 2 and degree 1, respectively. 
\end{construction}

 \begin{figure}[!htbp]
 \centering
    
          {\resizebox{0.2\textwidth}{!}{\begin{tikzpicture} [myv/.style={circle, draw, inner sep=5pt,line width=0.5mm},myv1/.style={circle, draw, inner sep=0pt,color=white},myv2/.style={rectangle, draw, dotted,  inner sep=2pt,line width = 0.5mm}]
 % \node [label=above: $\Large{r}$](z)  at (0,4) {};

  \node[myv] (a) at (0,0) {};
  \node[myv] (b) at (0,2) {};
  \node[myv] (c) at (2,0) {};

  \node[myv] (d) at (2,2) {};
  \node[myv] (e) at (4,0) {};
  \node[myv] (f) at (4,2) {};
  
    \draw[line width=0.5mm] (a) -- (b);
    \draw[line width=0.5mm] (b) -- (c);
    \draw[line width=0.5mm] (c) -- (d);
    \draw[line width=0.5mm] (d) -- (e);
    \draw[line width=0.5mm] (e) -- (f);
    \draw[line width=0.5mm] (e) -- (b);
    
\end{tikzpicture}}}  
          \caption{An example of a bipartite graph $G$ with bounded degree 3}
          \label{fig:bipartite}
 \end{figure}  

An example of Construction~\ref{cons:cubic} corresponding to the bipartite graph $G$ with bounded degree 3 given in Figure~\ref{fig:bipartite} is shown in Figure~\ref{fig:Cubic construction}. 
%Let $W$ be a gadget constructed in Construction~\ref{cons:cubic}. Let $H$ be a subgraph of $W$ induced by the vertices in two complete bipartite graphs $CB_1$ and $CB_2$, and the two leaves $b$ and $c$ of the star graph $S$.  Let $H'=H-\{b,c\}$. 
The following observation is true for the gadget $W$ in Construction~\ref{cons:cubic}.

 \begin{figure}[!htbp]
      
          \centering
          \textbf{{\resizebox{0.5\textwidth}{!}{\begin{tikzpicture} [myv/.style={circle, draw, inner sep=8pt,line width=0.8mm},myv1/.style={circle, draw, inner sep=0pt,color=white},myv2/.style={rectangle, draw, dotted,  inner sep=2pt,line width = 0.8mm}]
 % \node [label=above: $\Large{r}$](z)  at (0,4) {};

  \node[myv] (a) at (-4,1) {};
  \node[myv] (b) at (-4,3) {};
  \node[myv] (c) at (2,1) {};
  \node[myv] (d) at (2,3) {};
 % \node[circle, draw, inner sep=5pt] (d) at (4,2) {};
  \node[myv] (e) at (8,1) {};
  \node[myv] (f) at (8,3) {};
  
    \draw[line width=0.8mm]  (a) -- (b);
    \draw[line width=0.8mm] (b) -- (c);
    \draw[line width=0.8mm] (c) -- (d);
    \draw[line width=0.8mm] (d) -- (e);
    \draw[line width=0.8mm] (e) -- (f);
    \draw[line width=0.8mm] (e) -- (b);

    \node (g) at (-4,-4) { {\input{cubic_gadget1}}};
    \node[myv2] [fit=  (g), inner xsep=0.5ex, inner ysep=0.5ex] {};
    \node  (h) at (2,-4) {\input{cubic_gadget1}};
    \node[myv2]   [fit=  (h), inner xsep=0.5ex, inner ysep=0.5ex] {};
    \node  (i) at (8,-4) {\input{cubic_gadget1}};
    \node[myv2]   [fit=  (i), inner xsep=0.5ex, inner ysep=0.5ex] {};
    \node  (j) at (-4,8) {\input{cubic_gadget2}};
    \node[myv2]   [fit=  (j), inner xsep=0.5ex, inner ysep=0.5ex] {};
    \node   (k) at (2,8) {\input{cubic_gadget2}};
    \node[myv2]   [fit=  (k), inner xsep=0.5ex, inner ysep=0.5ex] {};
    \node   (l) at (8,8) {\input{cubic_gadget2}};
    \node[myv2]   [fit=  (l), inner xsep=0.5ex, inner ysep=0.5ex] {};

   \draw[line width=0.8mm] (-4, -1.25) -- (-4,0.7);
   \draw[line width=0.8mm] (2,-1.15) -- (-3.65,1.0);
   \draw[line width=0.8mm] (2.25,0.75) -- (8, -1.15);
   \draw[line width=0.8mm] (-4,5.15) -- (2,3.35);
   \draw[line width=0.8mm] (2, 5.15) -- (8,3.35);
   \draw[line width=0.8mm] (8,3.35) -- (8,5.2);

\end{tikzpicture}}} } 
          \caption{An example of the resultant graph $G_c$ obtained from Construction~\ref{cons:cubic} corresponding to the bipartite graph $G$ with bounded degree 3 in Figure~\ref{fig:bipartite}. The dotted rectangles represent the gadget $W$ used for the construction. }
          \label{fig:cubic}

    \label{fig:Cubic construction}
\end{figure}
\begin{observation}
    \label{obs:Hcoloring_cubic} 
     Let $W$ be the gadget,  and  
    $H$ as well as $H'$ be the subgraphs of $W$ as defined in Construction~\ref{cons:cubic}. In any $cd$-coloring of $H$, the vertices in $H'$ require at least 4 colors. %(Refer Figure~\ref{fig:cubic_gadget} for $H$ and $H'$). 
     Moreover, there exists a $cd$-coloring of $W$ using exactly 4 colors.
\end{observation}

\begin{proof}
    The graph $H'$ contains an induced cycle on 6 vertices, $C_6$ (for instance, the graph induced by the vertices $d,i,j,f,k,h$ as shown in Figure~\ref{fig:cubic_gadget}). It is easy to see that \XCD($C_6$)=4. Now consider the subgraph $H'$, which is obtained by 
     introducing new vertices adjacent to some of the vertices of the $C_6$.  Since none of the $3K_1$s present in the $C_6$ is dominated by any of the remaining vertices in $H'$, by Observation~\ref{obs:cdcolor_reduce} it is clear that the number of colors needed for any $cd$-coloring of $H'$ is at least 4.  Further, note that all the independent sets in $H'$, which are dominated by the vertices $b$ or $c$, are already dominated by some vertex in $H'$. Therefore, again by Observation~\ref{obs:cdcolor_reduce}, it is clear that the number of colors needed for $V(H')$ in any $cd$-coloring of $H$ is at least 4. 
     Since the root vertex $a$ of the gadget is not adjacent to any of the vertices in $H'$, the presence of $a$ will not reduce the number of colors needed for $cd$-coloring of the subgraph $H'$ of the corresponding gadget $W$. Hence, the number of colors needed for $cd$-coloring of $W$ is at least $4$. 
    Now it is easy to see that a 4-$cd$-coloring of any $C_6$ in $H'$ can be extended to a 4-$cd$-coloring of the corresponding gadget $W$ (by using the leaves of star graph $S$, and one of the vertices from each $C_i$, for $i=\{1,2\}$ as the dominating vertices of the color classes).
\end{proof}

The following lemma is based on the Construction~\ref{cons:cubic}, and as a consequence, we have Theorem~\ref{thm:cubic}.

\begin{lemma}
    \label{lem:cubic}
    Let $G$ be a bipartite graph with bounded degree 3 and $G_c$ be the graph obtained from Construction~\ref{cons:cubic}. Then \raisebox{2pt}{$\gamma$}$_t(G)=k$ if and only if $\XCD(G_c)=k+4x+8y$, where  $x$ and $y$  denote the number of vertices in $G$ having degree 2 and degree 1, respectively. 
\end{lemma}

\begin{proof}
    Let \raisebox{2pt}{$\gamma$}$_t(G)=k$. Then by Proposition~\ref{pro:triangle-free}, since $G$ is triangle-free, we have $\XCD(G)=k$. We claim that $\XCD(G_c)=k+4x+8y$.  Recall that the graph $G_c$ contains $x+2y$ gadgets, as there are $x$ number of degree two vertices, and $y$ number of degree one vertices in $G$. By Observation~\ref{obs:Hcoloring_cubic}, it is clear that these gadgets require exactly $4(x+2y)$ new colors, which are not used in $G$ (as none of these colors can be reused for coloring the vertices in $V(G_c)\setminus V(W)$).  Since \XCD($G$)=$k$, we then have \XCD($G_c$)=$k+4x+8y$.

    Now, for the reverse direction, let $\XCD(G_c)=m$. Then, we claim that \raisebox{2pt}{$\gamma$}$_t(G)=k$, where $k=m-(4x+8y)$. 
    Note that for each gadget $W$, the subgraph $W- \{a\}$ is disjoint from $G$. %, as only the vertex $a$ in $W$ is adjacent to exactly one vertex  $u$ in $G$.
    Hence, the color class dominated by $a$ can include only one vertex $u$ in $G$.  By Observation~\ref{obs:Hcoloring_cubic}, in any $cd$-coloring of $G_c$, the vertex set $H'$ of any gadget $W$ needs at least 4 colors, and none of these colors can be reused for coloring the vertices in $V(G_c)\setminus V(W)$. Again, by Observation~\ref{obs:Hcoloring_cubic}, the same 4 colors used to color the vertices in $H'$ can be reused to color the entire vertices of $W$. Therefore, we now have a $cd$-coloring of $G_c$ using $m$ colors such that no colors used for the vertices in gadgets are used to color any vertex in $G$. Thus, $\XCD(G)=m-4(x+2y)=k$. Then by Proposition~\ref{pro:triangle-free}, \raisebox{2pt}{$\gamma$}$_t(G)=k$.
\end{proof}

\begin{theorem}
    \label{thm:cubic}
    \CDC\ is \NPC\ on triangle free cubic graphs. \TETHS.
\end{theorem}

\begin{proof}
    Notice that $G_c$ has $|V(G)|+ 11(x+2y)$ vertices, where $x$ and $y$ are the number of vertices in $G$ having degree 2 and degree 1, respectively. We know that there is a linear reduction from \textsc{Total Domination} on bipartite graphs with bounded degree 3 to \CDC\ on triangle-free cubic graphs due to Lemma~\ref{lem:cubic}. Thus, by using Proposition~\ref{pro:bipartite}, we are done.
\end{proof}

Now, we generalize  Construction~\ref{cons:cubic} to prove the hardness of \CDC\ on triangle-free $d$-regular graphs, for any constant $d\geq 4$. The following construction is used to prove the hardness of \CDC\ on triangle-free $d$-regular graphs, for any odd integer $d\geq 5$. 
 \begin{construction}
    \label{cons:d-regular odd}
     Let $G$ be a triangle-free $(d-2)$-regular graph, for any odd integer $d\geq 5$. 
     We construct a graph $G_c$ from $G$ %using a gadget $W$ 
    in the following way: 
%    \vspace{-0.25cm}
    \begin{itemize}
        \item First we construct a gadget $W$ as follows:  introduce a star graph $K_{1,d-1}$, which we denote by $S$, having the vertex $a$ as the root vertex. Note that we also call `$a$' as the root vertex of $W$.
        Further introduce $2(d-1)^2$ vertices which induces $d-1$ disjoint copies of complete bipartite graphs $K_{d-1,d-1}$, namely $C_{1},  C_{2}\ldots,C_{d-1}$, respectively.      
        The adjacency between these complete bipartite graphs $C_i=(A_i,B_i)$, for $1\leq i\leq (d-1)$, and the star graph $S$ is in such a way that the vertices of $A_i$ of the complete bipartite graph $C_i$ is adjacent to the leaf vertex, say, $v_i$  of $S$.
        Furthermore, for each odd integer $i\leq (d-2)$, each vertex in the partite set $B_i$ of $C_i$ is adjacent to exactly one vertex of the partite set $B_{i+1}$ of $C_{i+1}$ (for instance refer  Figure~\ref{fig:d-regular_gadget odd}, when $d=5$). Now, by $H_i'$, we denote the subgraph of $W$ induced by the vertices in the two complete bipartite graphs $C_i$ and $C_{i+1}$, for an odd integer $i\leq d-2$.  Then, by $H_i$, we denote the subgraph of $W$ induced by the vertices in $H_i'$ and the two leaves $v_i$ and $v_{i+1}$ of the star graph $S$ which are adjacent to the partitions $A_i$ and  $A_{i+1}$ of $C_i$ and $C_{i+1}$, respectively. Note that $H_i'=H_i-\{v_i,v_{i+1}\}$. The gadget $W$ contains $(2d^2-3d+2)$ vertices.  
     
     \item Now the graph $G_c$ is constructed from $G$ using $W$ in such a way that for each vertex $u$ in $G$, introduce two copies of $W$ such that root vertices of both the copies of $W$ is attached to $u$. 
     \end{itemize}
     Note that  $G_c-G$ contains $2n(2d^2-3d+2)$ vertices, where $n$ is the number of vertices in $G$.
\end{construction}

An example of the gadget $W$ in Construction~\ref{cons:d-regular odd} is shown in Figure~\ref{fig:d-regular_gadget odd}. We have the following observation for the gadget $W$ of Construction~\ref{cons:d-regular odd}. 

\begin{figure}[ht]
  \centering
    \centering
    \resizebox{1\textwidth}{!}{\begin{tikzpicture}[myv/.style={circle, draw, inner sep=0pt},myv1/.style={circle, draw, inner sep=0pt,color=white},myv2/.style={rectangle, dotted, draw,inner sep=0pt,line width = 0.3mm},myv3/.style={circle, draw, inner sep=6pt},myv4/.style={rectangle, dashed, draw,inner sep=0pt,line width = 0.3mm}]
 % \node [label=above: $\Large{r}$](z)  at (0,4) {};

  \node[myv3] (a) at (8.5,6) {\huge{$a$}};
  \node[myv3] [fill=black] (b) at (-4,2) {\huge{\textcolor{white}{$b$}}};
  \node[myv3] [fill=black] (c) at (4,2) {\huge{\textcolor{white}{$c$}}};
  \node[myv3]  [fill=black](b1) at (12,2) {};
  \node[myv3]  [fill=black](c1) at (20,2) {};

  \node[myv]  [fill=black] (d1) at (-7,0) {\huge{\textcolor{white}{$d1$}}};
  \node[myv] (d2) at (-5,0) {\huge{$d2$}};
  \node[myv] (d3) at (-3,0) {\huge{$d3$}};
  \node[myv] (d4) at (-1,0) {\huge{$d4$}};

  \node[myv]  [fill=black](f1) at (1,0) {\huge{\textcolor{white}{$f1$}}};
  \node[myv] (f2) at (3,0) {\huge{$f2$}};
  \node[myv] (f3) at (5,0) {\huge{$f3$}};
  \node[myv] (f4) at (7,0) {\huge{$f4$}};
 
  \node[myv] (e1) at (-7,-2) {\huge{{$e1$}}};
  \node[myv] (e2) at (-5,-2) {\huge{$e2$}};
  \node[myv] (e3) at (-3,-2) {\huge{$e3$}};
  \node[myv] (e4) at (-1,-2) {\huge{$e4$}};
  \node[myv] (g1) at (1,-2) {\huge{{$g1$}}};
  \node[myv] (g2) at (3,-2) {\huge{$g2$}};
  \node[myv] (g3) at (5,-2) {\huge{$g3$}};
  \node[myv] (g4) at (7,-2) {\huge{$g4$}};

  \node[myv3]  [fill=black](h1) at (9,0) {};
  \node[myv3] (h2) at (11,0) {};
  \node[myv3] (h3) at (13,0) {};
  \node[myv3] (h4) at (15,0) {};

  \node[myv3]  [fill=black](j1) at (17,0) {};
  \node[myv3] (j2) at (19,0) {};
  \node[myv3] (j3) at (21,0) {};
  \node[myv3] (j4) at (23,0) {};
 
  \node[myv3] (i1) at (9,-2) {};
  \node[myv3] (i2) at (11,-2) {};
  \node[myv3] (i3) at (13,-2) {};
  \node[myv3] (i4) at (15,-2) {};
  \node[myv3] (k1) at (17,-2) {};
  \node[myv3] (k2) at (19,-2) {};
  \node[myv3] (k3) at (21,-2) {};
  \node[myv3] (k4) at (23,-2) {};

   \node[myv1] (l)  at (3,-3.5) {};
   \node[myv1] (m)  at (3,-4) {};
  % \node[myv1] (n)  at (3,-4) {};
  
    \draw (a) -- (b);
    \draw (a) -- (c);
    \draw (a) -- (b1);
    \draw (a) -- (c1);
    \draw (b) -- (d1);
    \draw (b) -- (d2);
    \draw (b) -- (d3);
    \draw (b) -- (d4);
    \draw (c) -- (f1);
    \draw (c) -- (f2);
    \draw (c) -- (f3);
    \draw (c) -- (f4);
    
    \draw (d1) -- (e1);
    \draw (d1) -- (e2);
    \draw (d1) -- (e3);
    \draw (d1) -- (e4);
    \draw (d2) -- (e1);
    \draw (d2) -- (e2);
    \draw (d2) -- (e3);
    \draw (d2) -- (e4);

    \draw (d3) -- (e1);
    \draw (d3) -- (e2);
    \draw (d3) -- (e3);
    \draw (d3) -- (e4);
    \draw (d4) -- (e1);
    \draw (d4) -- (e2);
    \draw (d4) -- (e3);
    \draw (d4) -- (e4);

    \draw (f1) -- (g1);
    \draw (f1) -- (g2);
    \draw (f1) -- (g3);
    \draw (f1) -- (g4);
    \draw (f2) -- (g1);
    \draw (f2) -- (g2);
    \draw (f2) -- (g3);
    \draw (f2) -- (g4);

    \draw (f3) -- (g1);
    \draw (f3) -- (g2);
    \draw (f3) -- (g3);
    \draw (f3) -- (g4);
    \draw (f4) -- (g1);
    \draw (f4) -- (g2);
    \draw (f4) -- (g3);
    \draw (f4) -- (g4);

    \draw (h1) -- (i1);
    \draw (h1) -- (i2);
    \draw (h1) -- (i3);
    \draw (h1) -- (i4);
    \draw (h2) -- (i1);
    \draw (h2) -- (i2);
    \draw (h2) -- (i3);
    \draw (h2) -- (i4);

    \draw (h3) -- (i1);
    \draw (h3) -- (i2);
    \draw (h3) -- (i3);
    \draw (h3) -- (i4);
    \draw (h4) -- (i1);
    \draw (h4) -- (i2);
    \draw (h4) -- (i3);
    \draw (h4) -- (i4);

    \draw (j1) -- (k1);
    \draw (j1) -- (k2);
    \draw (j1) -- (k3);
    \draw (j1) -- (k4);
    \draw (j2) -- (k1);
    \draw (j2) -- (k2);
    \draw (j2) -- (k3);
    \draw (j2) -- (k4);

    \draw (j3) -- (k1);
    \draw (j3) -- (k2);
    \draw (j3) -- (k3);
    \draw (j3) -- (k4);
    \draw (j4) -- (k1);
    \draw (j4) -- (k2);
    \draw (j4) -- (k3);
    \draw (j4) -- (k4);

    \draw (b1) -- (h1);
    \draw (b1) -- (h2);
    \draw (b1) -- (h3);
    \draw (b1) -- (h4);
    \draw (c1) -- (j1);
    \draw (c1) -- (j2);
    \draw (c1) -- (j3);
    \draw (c1) -- (j4);

    \draw (e1) to [bend right=20] (g4);
    \draw (e2) to [bend right=20] (g3);
    \draw (e3) to [bend right=20] (g2);
    \draw (e4) to [bend right=20] (g1);
    \draw (i1) to [bend right=20] (k4);
    \draw (i2) to [bend right=20] (k3);
    \draw (i3) to [bend right=20] (k2);
    \draw (i4) to [bend right=20] (k1);

   \node[myv4][fit=  (d1) (e1) (g4)  (l),  inner xsep=1.5ex, inner ysep=2.5ex, label=above:{\huge{$H_1'$}}] {};
   \node[myv2][fit=(b) (e1) (g4) (m),  inner xsep=5.5ex, inner ysep= 5.75ex, label=above:\huge{$H_1$}] {}; 
 
\end{tikzpicture}}
    \caption{An example of the  gadget $W$ used in Construction~\ref{cons:d-regular odd}, for $d=5$  such that the shaded vertices  correspond to the dominating vertices of the color classes in a valid $cd$-coloring of $W$.} 
    \label{fig:d-regular_gadget odd}
  \end{figure}

\begin{observation}
    \label{obs:Hcoloring_regular_odd} 
     Let $W$ be a gadget, and  
    $H_i$ as well as $H_i'$ be the subgraphs of $W$ as defined in Construction~\ref{cons:d-regular odd}. 
     In any $cd$-coloring of $H_i$, the vertices in $H_i'$ require at least 4 colors. 
     Moreover, there exists a $cd$-coloring of $W$ using exactly $2(d-1)$ colors.
\end{observation}

\begin{proof}
    The graph induced by $H_i'$ contains a $C_6$ (for instance: induced by the vertices $d1,e1,g4,f4,g3,$ $e2$ as shown in Figure~\ref{fig:d-regular_gadget odd}). 
    It is clear that \XCD($C_6$)=4. 
    Now consider the subgraph  $H_i'$. It is obtained by 
    the introduction of new vertices adjacent to some vertices of the $C_6$.   Since none of the $3K_1$s present in the $C_6$ is dominated by any of the remaining vertices in $H_i'$, by Observation~\ref{obs:cdcolor_reduce}, it is clear that the number of colors needed for any $cd$-coloring of $H_i'$ is at least 4. Further, note that all the independent sets in $H_i'$, which are dominated by the leaf vertices $v_i$ or $v_{i+1}$ of the star graph $S$ (for instance the vertices $b$ and $c$ shown in Figure~\ref{fig:d-regular_gadget odd}), are already dominated by some vertex in $H_i'$. Therefore, again, by Observation~\ref{obs:cdcolor_reduce}, it is clear that the number of colors needed for $V(H_i')$ in any $cd$-coloring of $H_i$ is at least 4. 
    Since the root vertex $a$ of the gadget is not adjacent to any of the vertices in $H_i'$, the presence of $a$ will not reduce the number of colors needed for $cd$-coloring of the subgraph $H_i'$ of the gadget $W$. Hence, the number of colors needed for $cd$-coloring of $H_i$ is at least $4$. Therefore, the number of colors needed for $cd$-coloring of $W$ is at least $2(d-1)$, as there are $(d-1)/2$ copies of $H_i$ is present in $W$. 
    Now it is easy to see that a 4-$cd$-coloring of a $C_6$ in each $H_i'$ can be extended to a $2(d-1)$-$cd$-coloring of the gadget $W$ (by using the leaves of star graph $S$, and one of the vertices from the partition $A_i$ of each $C_i$, for $1\leq i\leq d-1$, as the dominating vertices of the color classes). 
\end{proof}

The following lemma is based on the Construction~\ref{cons:d-regular odd}.

\begin{lemma}
    \label{lem:d-regular odd}
    For each odd integer $d\geq 5$, let $G$ be a triangle-free ($d-2$)-regular graph, having $n$ vertices. Then $\XCD(G)=k$ if and only if $\XCD(G_c)=k+4n(d-1)$. 
\end{lemma}

\begin{proof}
    Let $\XCD(G)=k$. We claim that $\XCD(G_c)=k+4n(d-1)$.  Recall that the graph $G_c$ contains $2n$ gadgets as there are $n$ vertices in $G$ each having degree $d-2$. By Observation~\ref{obs:Hcoloring_regular_odd}, it is clear that any $cd$-coloring of these gadgets needs a total of at least $4n(d-1)$ colors. Since \XCD($G$)=$k$, we then have \XCD($G_c$)=$k+4n(d-1)$. 

    Now, for the reverse direction, let $\XCD(G_c)=m$. Then we claim that $\XCD(G)=k$, where $k=m-4n(d-1)$. 
    Note that for each gadget $W$, the subgraph $W- \{a\}$ is disjoint from $G$, as only $a$ in $W$ is adjacent to exactly one vertex  $u$ in $G$. Hence, the color class dominated by $a$ can include only one vertex $u$ in $G$. 
     By Observation~\ref{obs:Hcoloring_regular_odd}, in any $cd$-coloring of $G_c$, the vertex set of each copy of $H'$ of any gadget $W$ needs at least 4 colors, and none of these colors can be reused for coloring the vertices in $V(G_c)\setminus V(W)$. Again, by Observation~\ref{obs:Hcoloring_regular_odd}, the same 4 colors used to color the vertices in $H_i'$ can be reused to color the entire vertices of $H_i$. Hence, the 2($d-1$) colors used for $(d-1)/2$ copies of $H_i$ in $W$ can be reused to color the entire vertices of $W$. Therefore, we now have a $cd$-coloring of $G_c$ using $m$ colors such that no colors used for the vertices in gadgets are used to color any vertex in $G$. Thus, $\XCD(G)=m-4n(d-1)=k$, as there are $2n$ copies of $W$s in $G_c$.  
\end{proof}

The following construction is used to prove the hardness of \CDC\ on triangle-free $d$-regular graphs, for any even integer $d\geq 4$.  
\begin{construction}
    \label{cons:d-regular even}
     Let $G$ be a triangle-free $(d-1)$-regular graph, for any even integer $d\geq 4$.  We construct a graph $G_c$ from $G$ %using a gadget $W$ 
    in the following way: 
   % \vspace{-0.25cm}
    \begin{itemize}
        \item First we construct a gadget $W$ as follows: introduce two sets of $2d^2-5d+3$ vertices which induce two subgraphs $W_1$ and $W_2$ in $W$. 
        The adjacency among the vertices in $W_1$ (respectively $W_2$) is in such  a way that the $d-1$ vertices of $W_1$ (respectively $W_2$) induces a star graph $K_{1,d-2}$,  which we denote by $S_1$ (respectively $S_2$), having the vertex $a$ in $W_1$ (respectively the vertex $b$ in $W_2$) as the root vertex. Note that the vertices $a$ and $b$ are adjacent. 
         
         \item Further in each set $W_i$, for $i\in \{1,2\}$, the remaining $(2d^2-6d+4)$ vertices  induces $d-2$ disjoint copies of complete bipartite graphs $K_{d-1,d-1}$, namely $C_{i1},  C_{i2}\ldots,$ $C_{i(d-2)}$ respectively.      
     The adjacency between these complete bipartite graphs $C_{ij}$, for $i\in\{1,2\}$ and $1\leq j\leq d-2$ and star graph $S_i$,  is in such a way that the vertices of one of the partition $A_{ij}$  of a complete bipartite graph $C_{ij}$ is adjacent to the leaf $v_j$  of  $S_i$.       
     Furthermore, for each odd integer $j\leq d-3$, each vertex of the partite set $B_{ij}$ of $C_{ij}$ is adjacent to exactly one vertex of  $B_{i(j+1)}$ of $C_{i(j+1)}$  (for instance, refer Figure~\ref{fig:d-regular_gadget even}, when $d=4$). 
     Now, by $H_{ij}'$ we denote the subgraph of $W_i$ induced by the vertices in the two complete bipartite graphs $C_{ij}$ and $C_{i(j+1)}$, for an odd integer $j\leq d-3$.  Then, by $H_{ij}$ we denote the subgraph of $W_i$ induced by the vertices in $H_{ij}'$ and  the two leaves $v_j$ and $v_{j+1}$ of the star graph $S_i$ which are adjacent to the partitions $A_{ij}$ and  $A_{i(j+1)}$ of $C_{ij}$ and $C_{i(j+1)}$, respectively. Note that $H_{ij}'=H_{ij}-\{v_j,v_{j+1}\}$.
     
     \item Now the graph $G_c$ is constructed from $G$ using $W$ in the following way: consider an arbitrary pair-wise ordering $(v_1,v_2), (v_3,v_4)\ldots,(v_{n-1}, v_n)$ of vertices, in $G$.  Note that such a pairing is possible, since $n$ is even (as $d-1$ is odd). 
     For a pair of vertices $(v_j,v_{j+1})$, for odd integer $j\leq n-1$ in this ordering, introduce a gadget $W$ such that the vertex $a$ (respectively $b$) of $W$ is adjacent to $v_j$ (respectively $v_{j+1}$) of $G$.
     \end{itemize}
     The graph $G_c-G$ contains $n(2d^2-5d+3)$ vertices, where $n$ is the number of vertices in $G$.
     
\end{construction}

An example of the gadget $W$ in Construction~\ref{cons:d-regular even} is shown in Figure~\ref{fig:d-regular_gadget even}. The following observation is true for the gadget $W$ of Construction~\ref{cons:d-regular even}. 

\begin{figure}[ht]
  \centering
    \centering
    \resizebox{0.8\textwidth}{!}{\begin{tikzpicture}[myv/.style={circle, draw, inner sep=0pt},myv1/.style={circle, draw, inner sep=0pt,color=white},myv2/.style={rectangle, draw,inner sep=0pt,line width = 0.5mm},myv3/.style={circle, draw, inner sep=6pt},myv4/.style={circle, inner sep=0pt}]
 % \node [label=above: $\Large{r}$](z)  at (0,4) {};

  \node[myv3] (a) at (-2,6) {\huge{$a$}};
  \node[myv3] (a1) at (11,6) {\huge{$b$}};
  \node[myv3]  [fill=black](b) at (-5,2) {\huge{\textcolor{white}{$c$}}};
  \node[myv3]  [fill=black](c) at (1,2) {\huge{\textcolor{white}{$d$}}};
  \node[myv3]  [fill=black](b1) at (8,2) {};
  \node[myv3]  [fill=black](c1) at (14,2) {};

  \node[myv]  [fill=black](d1) at (-7,0) {\huge{\textcolor{white}{$d1$}}};
  \node[myv] (d2) at (-5,0) {\huge{$d2$}};
  \node[myv] (d3) at (-3,0) {\huge{$d3$}};

  \node[myv]  [fill=black](f1) at (-1,0) {\huge{\textcolor{white}{$f1$}}};
  \node[myv] (f2) at (1,0) {\huge{$f2$}};
  \node[myv] (f3) at (3,0) {\huge{$f3$}};
   
  \node[myv] (e1) at (-7,-2) {\huge{$e1$}};
  \node[myv] (e2) at (-5,-2) {\huge{$e2$}};
  \node[myv] (e3) at (-3,-2) {\huge{$e3$}};
 
  \node[myv] (g1) at (-1,-2) {\huge{$g1$}};
  \node[myv] (g2) at (1,-2) {\huge{$g2$}};
  \node[myv] (g3) at (3,-2) {\huge{$g3$}};

  \node[myv3]  [fill=black](h1) at (6,0) {};
  \node[myv3] (h2) at (8,0) {};
  \node[myv3] (h3) at (10,0) {};
  
  \node[myv3]  [fill=black](j1) at (12,0) {};
  \node[myv3] (j2) at (14,0) {};
  \node[myv3] (j3) at (16,0) {};
 
  \node[myv3] (i1) at (6,-2) {};
  \node[myv3] (i2) at (8,-2) {};
  \node[myv3] (i3) at (10,-2) {};
 
  \node[myv3] (k1) at (12,-2) {};
  \node[myv3] (k2) at (14,-2) {};
  \node[myv3] (k3) at (16,-2) {};
 
   \node[myv1] (l)  at (4,-3.5) {};
   \node[myv1] (m)  at (4,-4) {};
   
   \node[myv1] (n)  at (14,-4.5) {};
   
    \draw (a) -- (a1);
    \draw (a) -- (b);
    \draw (a) -- (c);
    \draw (a1) -- (b1);
    \draw (a1) -- (c1);
    \draw (b) -- (d1);
    \draw (b) -- (d2);
    \draw (b) -- (d3);

    \draw (c) -- (f1);
    \draw (c) -- (f2);
    \draw (c) -- (f3);
    
    \draw (d1) -- (e1);
    \draw (d1) -- (e2);
    \draw (d1) -- (e3);
    
    \draw (d2) -- (e1);
    \draw (d2) -- (e2);
    \draw (d2) -- (e3);
  
    \draw (d3) -- (e1);
    \draw (d3) -- (e2);
    \draw (d3) -- (e3);
  
    \draw (f1) -- (g1);
    \draw (f1) -- (g2);
    \draw (f1) -- (g3);
   
    \draw (f2) -- (g1);
    \draw (f2) -- (g2);
    \draw (f2) -- (g3);
   
    \draw (f3) -- (g1);
    \draw (f3) -- (g2);
    \draw (f3) -- (g3);
   
    \draw (h1) -- (i1);
    \draw (h1) -- (i2);
    \draw (h1) -- (i3);
   
    \draw (h2) -- (i1);
    \draw (h2) -- (i2);
    \draw (h2) -- (i3);
   
    \draw (h3) -- (i1);
    \draw (h3) -- (i2);
    \draw (h3) -- (i3);
   
    \draw (j1) -- (k1);
    \draw (j1) -- (k2);
    \draw (j1) -- (k3);
   
    \draw (j2) -- (k1);
    \draw (j2) -- (k2);
    \draw (j2) -- (k3);
  
    \draw (j3) -- (k1);
    \draw (j3) -- (k2);
    \draw (j3) -- (k3);
    
    \draw (b1) -- (h1);
    \draw (b1) -- (h2);
    \draw (b1) -- (h3);
  
    \draw (c1) -- (j1);
    \draw (c1) -- (j2);
    \draw (c1) -- (j3);
     
    \draw (e1) to [bend right=20] (g3);
    \draw (e2) to [bend right=20] (g2);
    \draw (e3) to [bend right=20] (g1);
   
    \draw (i1) to [bend right=20] (k3);
    \draw (i2) to [bend right=20] (k2);
    \draw (i3) to [bend right=20] (k1);

    \node[myv2][dotted](H)[fit=(b) (e1) (g3) (m),  inner xsep=3.5ex, inner ysep=3.00ex] {}; 
    
    \node[myv4] at (-7,3.75) {\huge{$H_{11}$}};
    
    \node[myv2][dashed][fit=  (d1) (e1) (g3)  (l),  inner xsep=0.6ex, inner ysep=2.5ex] {};
    
     \node[myv4] at (-7,1.75) {\huge{$H_{11}'$}};

     \node[myv2][dashed][fit=  (a) (b) (H)  ,  inner xsep=1.5ex, inner ysep=2.5ex] {};

      \node[myv4] at (-7,7.75) {\huge{$W_1$}};

     \node[myv2][dashed][fit=  (a1) (k3) (h1) (n),  inner xsep=1.5ex, inner ysep=2.5ex] {};

      \node[myv4] at (7,7.75) {\huge{$W_2$}};

\end{tikzpicture}}
    \caption{An example of the  gadget $W$ used in Construction~\ref{cons:d-regular even}, for $d=4$, such that shaded vertices correspond to the dominating vertices of the color classes in a valid $cd$-coloring of $W$.} 
    \label{fig:d-regular_gadget even}
  \end{figure}

\begin{observation}
    \label{obs:Hcoloring_regular_even}
    Let $W$ be a gadget, and  
    $H_{ij}$ as well as $H_{ij}'$, for $i\in\{1,2\}$ and  an odd integer $j\leq d-3$,  be the subgraphs of $W$ as defined in Construction~\ref{cons:d-regular even}. 
     In any $cd$-coloring of $H_{ij}$, the vertices in $H_{ij}'$ require at least 4 colors. 
     Moreover there exists a $cd$-coloring of $W$ using exactly  $4(d-2)$ colors.
\end{observation}
\begin{proof}
    %Let $H_{ij}$ be a subgraph of $W_{i}$, for $i\in \{1,2\}$ in $W$. 
    The graph induced by $H_{ij}'$ contains a $C_6$ (for instance: induced by the vertices $d1,e1,g3,f3,g2,$ $e2$ as shown in Figure~\ref{fig:d-regular_gadget even}). 
    It is clear that \XCD($C_6$)=4. 
    Now consider the subgraph $H_{ij}'$. It is obtained by 
    the introduction of new vertices adjacent to some vertices of the $C_6$.  Since none of the $3K_1$s present in the $C_6$ is dominated by any of the remaining vertices in $H_{ij}'$, by Observation~\ref{obs:cdcolor_reduce}, it is clear that the number of colors needed for any $cd$-coloring of $H_{ij}'$ is at least 4. 
    Further, note that all the independent sets in $H_{ij}'$, which are dominated by $v_j$ or $v_{j+1}$ of $S_i$, for $i\in\{1,2\}$, are already dominated by some vertex in $H_{ij}'$. Therefore, again by Observation~\ref{obs:cdcolor_reduce}, it is clear that the number of colors needed for $V(H_{ij}')$ in any $cd$-coloring of $H_{ij}$ is at least 4. 
    Since the vertex $a$ (respectively $b$) of $W_1$ (respectively $W_2$) is not adjacent to any of the vertices in $H_{1j}'$ (respectively $H_{2j}'$), the presence of $a$ (respectively $b$) will not reduce the number of colors needed for $cd$-coloring of the subgraph  $H_{1j}'$ (respectively $H_{2j}'$) of the gadget $W$. Hence, the number of colors needed for a $cd$-coloring of  $H_{1j}$ (respectively $H_{2j}$) is at least $4$. Therefore, the number of colors needed for a $cd$-coloring of $W_i$, for $i\in \{1,2\}$, is at least $2(d-2)$, as there are $(d-2)/2$ copies of $H_{ij}$ is present in $W_i$. %Similarly,  the number of colors needed for $cd$-coloring of $W_2$ is at least $2(d-2)$, as there are $(d-2)/2$ copies of $H$ is present in $W_2$. 
    Hence, it is clear that the number of colors needed for a $cd$-coloring of $W$ is at least $4(d-2)$.
    Now it is easy to see that a 4-$cd$-coloring of a $C_6$ in each $H_{ij}'$ can be extended to a 2$(d-2)-cd$-coloring of the subgraph $W_i$  for $i\in \{1,2\}$, (by using the leaves of star graph $S_i$, and one of the vertices from the partition $A_{ij}$ of each $C_{ij}$, for $1\leq j\leq (d-2)/2$, as the dominating vertices of the color classes). Now it is obvious that, these 2$(d-2)$-$cd$-coloring of the subgraph $W_1$ and $W_2$ can be extended to a 4$(d-2)$-$cd$-coloring of $W$. 
\end{proof}

The following lemma is based on the Construction~\ref{cons:d-regular even}.

\begin{lemma}
    \label{lem:d-regular even}
    For an even integer $d\geq 4$, let $G$ be a triangle-free $(d-1)$-regular graph having $n$ vertices. Then $\XCD(G)=k$ if and only if $\XCD(G_c)=k+2n(d-2)$. 
\end{lemma}

\begin{proof}
    Let $\XCD(G)=k$. We claim that $\XCD(G_c)=k+2n(d-2)$. Recall that the graph $G$ contains an even number of vertices as it is a $(d-1)$-regular graph, where $d-1$ is odd. Hence $G_c$ contains ${n}/{2}$ gadgets as there are $n$ vertices in $G$ each having degree $d-1$. By Observation~\ref{obs:Hcoloring_regular_even}, it is clear that these gadgets need at least $2n(d-2)$ colors. Since \XCD($G$)=$k$, we then have \XCD($G_c$)=$k+2n(d-2)$.

    Now, for the reverse direction, let $\XCD(G_c)=m$. Then we claim that $\XCD(G)=k$, where $k=m-2n(d-2)$.  Note that only the vertices $a$ and $b$ of each gadget $W$ are adjacent to some vertices  $u$ and $v$, respectively, in $G$. Hence, the color class dominated by $a$ (respectively $b$) can include only one vertex $u$ (respectively $v$) in $G$.   
    By Observation~\ref{obs:Hcoloring_regular_even}, in any $cd$-coloring of $G_c$, the vertex set of each copy of $H_{ij}'$, for $i\in \{1,2\}$ and an odd integer $j\leq {(d-2)}/{2}$, of any gadget $W$ needs at least 4 colors, and none of these colors can be reused for coloring the vertices in $V(G_c)\setminus V(W)$. Again, by Observation~\ref{obs:Hcoloring_regular_even}, the same 4 colors used to color the vertices in $H_{ij}'$ can be reused to color the entire vertices of $H$. Hence, the 4($d-2$) colors used for $(d-2)$ copies of $H_{ij}$ in $W$ can be reused to color the entire vertices of $W$. Therefore, we now have a $cd$-coloring of $G_c$ using $m$ colors such that no colors used for the vertices in gadgets are used to color any vertex in $G$. Thus, $\XCD(G)=m-2n(d-2)=k$, as there are ${n}/{2}$ copies of $W$s in $G_c$.  
\end{proof}

\noindent\textit{Proof of Theorem~\ref{thm:regular}:} Note that $\mid V(G_{c1})\mid =2n(2d^2-3d+2)$, where $G_{c1}$ is obtained as per construction~\ref{cons:d-regular odd}. Similarly, $\mid V(G_{c2})\mid =n(2d^2-5d+3)$, where $G_{c2}$ is obtained as per construction~\ref{cons:d-regular even}. Hence, both constructions are linear with respect to their size of
inputs, and hence, their associated reductions are also linear. 
Thus, we know that there is a linear reduction from triangle-free $(d-2)$-regular graphs to triangle-free $d$-regular graphs for each fixed odd integer $d\geq 5$ due to Lemma \ref{lem:d-regular odd}. 
We also know that there is a linear reduction from triangle-free $(d-1)$-regular graphs to triangle-free $d$-regular graphs for each fixed even integer $d\geq 4$ due to Lemma \ref{lem:d-regular even}. Therefore, we are done using Proposition~\ref{pro:bipartite} and Theorem~\ref{thm:cubic}. \qed
\section{Separated-Cluster and $cd$-coloring: $cd$-perfectness}
\label{sec:structural}

\label{subsec:cdprfect}
 In this section, we study the relationship of the two problems \SCP\ and \CDC\ and introduce the notion of \textit{cd-perfect graphs}, which is defined below.
\begin{definition}[\textit{cd}-perfect graphs]
    An undirected graph $G$ is said to be $cd$-perfect if, for any induced subgraph $H$ of $G$, we have $\XCD(H)=\SC(H)$.
\end{definition}

  Recall from Definition~\ref{def:aux} that given a graph $G$, the auxiliary graph $G^*$ is the graph having $V(G^*)=V(G)$ and $E(G^*)=E(G^2)\setminus E(G)$. i.e. $E(G^*)= \{uv:u,v\in V(G)$ and $d_G(u,v)=2\}$. First, we study the parameters $\XCD(G)$ and $\SC(G)$ of a graph $G$, by relating them to some well-known graph parameters of the auxiliary graph $G^*$. This simple reduction has various interesting consequences. Since the results in this section heavily rely on $G^*$, throughout this section, we use the following definition of \textit{seperated-cluster}: a set $S\subseteq V(G)$ is a seperated-cluster, if for any pair of vertices $u,v\in V(G)$, we have $d_G(u,v)\neq 2$.

 The following proposition trivially follows from the definition of $G^*$.
\begin{proposition} \label{pro:omega_s}
    For any graph $G$, we have $\SC (G)=\alpha(G^*)$.
\end{proposition}

We then note some observations that further lead us to some necessary and sufficient conditions for $cd$-perfectness. Even though the derivations of these conditions are apparently simple, they have some interesting consequences for obtaining the results in the latter part of the section.
\begin{observation} \label{obs:cliquecover}
    For any graph $G$, we have $\XCD (G)\geq k(G^*)$.
\end{observation}
\begin{proof}
    Let $\XCD(G)=l$, where $I_1,I_2,\ldots, I_l$ denote the corresponding color classes of $G$. Let $j\in \{1,2,\ldots,l\}$. Then by the definition of $cd$- coloring, we have that $I_j$ is an independent set in $G$ and there exists a vertex $v_j\in V(G)$ such that $I_j\subseteq N_G[v_j]$. This implies that for any pair of vertices $u,v\in V(I_j)$, we have $d_G(u,v)=2$. This further implies that $I_j$ is a clique in $G^*$, and therefore $\{I_1,I_2,\ldots, I_l\}$ is a clique cover of $G^*$. Thus we have $\XCD (G)\geq k(G^*)$.
\end{proof}

Note that the reverse inequality of Observation~\ref{obs:cliquecover} is not necessarily true. For instance, consider the graph $C_6$ (an induced cycle on 6 vertices). It is not difficult to verify that $\XCD(C_6)=4$. But as $C_6^*$ is a disjoint union of two triangles, we have $k(C_6^*)=2<\XCD(C_6)$. In Theorem~\ref{thm:suff_cdclique}, we prove a sufficient condition for a graph to satisfy the equality in Observation~\ref{obs:cliquecover}. First, we define certain graphs. For this, consider the graph $C_6$. Since $C_6$ is bipartite, observe that $V(C_6)$ can be partitioned into two independent sets, say $A$ and $B$, where $\mid A\mid~ =~\mid B\mid~ =3$. We denote by $C_6^1$, $C_6^2$, and $C_6^3$, the graphs obtained by respectively adding \textit{1 edge, 2 edges, and 3 edges} to exactly one of the partite sets $A$ or $B$ of $C_6$ (see Figure~\ref{fig:c6-free}). Let $\mathcal{H}=\{C_6, C_6^1$, $C_6^2,C_6^3\}$ (refer Figure~\ref{fig:c6-free}). Now, we note a useful structural observation for graphs in $\mathcal{H}$. 
\begin{observation}\label{obs:structureH}
    For each $H\in \mathcal{H}$, we have an independent set $S\subseteq V(H)$ with $\mid S\mid~ =3$ such that the vertices in $S$ form an asteroidal-triple in $H$.
\end{observation}

Now, in the following lemma, we prove a property satisfied by $G^*$, when $G$ is restricted to be an $\mathcal{H}$-free graph. 
\begin{lemma}\label{lem:neighborhood}
Let $G$ be an $\mathcal{H}$-free graph and $G^*$ its corresponding auxiliary graph. Let $K$ be any clique in $G^*$. Then there exists a vertex $v_k\in V(G)$ such that $K\subseteq N_G(v_k)$.
\end{lemma}
\begin{proof}
    We prove this by induction on $\mid K\mid $. Consider the base case where $\mid K\mid ~=2$. Let $K=\{x,y\}$. Note that $K$ is a clique in $G^*$. By the definition of $G^*$, $xy\in E(G^*)$ implies that $d_G(x,y)=2$. i.e. there exists a vertex $v_k$ such that $x,y\in N_G(v_k)$. As this proves the base case, we can, therefore, assume that $\mid K\mid\geq~ 3$. By the induction hypothesis, we can also assume that the condition in the lemma holds for every clique, say $K'$ in $G^*$ with $\mid K'\mid ~<~\mid K\mid$. Let $x_1,x_2,x_3\in K$ (this is possible, since $\mid K\mid~ \geq 3)$. For $i\in \{1,2,3\}$, let $K_i$ denote the clique $K\setminus \{x_i\}$. By the induction hypothesis, there exists a vertex, say $v_i\in V(G)$, such that $K_i\subseteq N_G(v_i)$. Note that $K$ is an independent set in $G$ (as $K$ is a clique in $G^*$), and therefore $v_i\notin K$ for any $i\in \{1,2,3\}$. Also, in particular, we have $v_ix_j\in E(G)$ for each $i,j\in \{1,2,3\}$, where $i\neq j$. Suppose that $v_ix_i\notin E(G)$ for each $i\in \{1,2,3\}$.  This implies that $v_1$, $v_2$, and $v_3$ are all distinct vertices in $G$. Consider the induced subgraph, $H=G[\{x_1,x_2,x_3,v_1,v_2,v_3\}]$. Since $\{x_1,x_2,x_3\}$ is an independent set in $G$, we then have, 
    
\begin{align}{
   \textit{H}=    \left\lbrace \begin{aligned}
 C_6, & \text{ if \hspace{0.1cm}}\mid E(G[\{v_1,v_2,v_3\}])\mid=0 \\
C_6^1, & \text{ if \hspace{0.1cm}} \mid E(G[\{v_1,v_2,v_3\}])\mid=1 \\
 C_6^2, & \text{ if \hspace{0.1cm}} \mid E(G[\{v_1,v_2,v_3\}])\mid=2 \\
 C_6^3, & \text{ if \hspace{0.1cm}} \mid E(G[\{v_1,v_2,v_3\}])\mid =3
\end{aligned} \right .}
\end{align}

This contradicts the fact that $G$ is $\mathcal{H}$-free. Therefore, we can conclude that $v_ix_i\in E(G)$ for some $i\in \{1,2,3\}$. Then, as $K=K_i\cup\{x_i\}$, the vertex $v_i$ in $G$ has the property that $K\subseteq N_G(v_i)$. This proves the lemma.
\end{proof}

\begin{theorem}\label{thm:suff_cdclique}
Let $G$ be an $\mathcal{H}$-free graph and $G^*$ its corresponding auxiliary graph. Then $\XCD (G)= k(G^*)$.
\end{theorem}

\begin{proof}
By Observation~\ref{obs:cliquecover}, it is enough to show that $\XCD (G)\leq k(G^*)$. Let $k(G^*)=t$, and let $K_1,K_2,\ldots, K_t$ be a minimum sized clique cover of $G^*$. Since $G$ is $\mathcal{H}$-free, for each $j\in \{1,2,\ldots,t\}$, by Lemma~\ref{lem:neighborhood}, we have that there exists a vertex $v_j\in V(G)$ such that $K_j\subseteq N_G(v_j)$. As each set $K_j$ is an independent set in $G$ (since $K_j$ is a clique in $G^*$), we can now infer that each set $K_j$ is a valid color class for a $cd$-coloring of $G$. Since $V(G)$ is the disjoint union of the sets, $K_1,K_2,\ldots, K_t$, we can therefore conclude that $\XCD (G)\leq k(G^*)$. This proves the theorem.
 \end{proof}
Now Theorem~\ref{thm:suff_cdsubclique}, which provides a sufficient condition for a graph $G$ to have $\XCD(G)=\SC(G)$, is  immediate from Proposition~\ref{pro:omega_s} and Theorem~\ref{thm:suff_cdclique}.  The following corollary that provides a sufficient condition for a graph to be $cd$-perfect can also be easily inferred from Theorem~\ref{thm:suff_cdclique}.

 \begin{corollary}\label{corr:suffcdperfect}
     Let $G$ be an $\mathcal{H}$-free graph. If for any induced subgraph $H$ of $G$, $k(H^*)=\alpha(H^*)$ then $G$ is $cd$-perfect. Consequently, if $G$ is $\mathcal{H}$-free and $H^*$ is perfect for any induced subgraph $H$ of $G$, then $G$ is $cd$-perfect.
 \end{corollary}
 
 It is known in the literature that if $G$ is a co-bipartite graph, then $\XCD(G)=\SC(G)$~\cite{ShaluKiruCDTreeCobipar21}. As a consequence of Theorem~\ref{thm:suff_cdsubclique}, in the following corollary, we now have a simple and shorter proof for the same. 
 \begin{corollary}\label{corr:cobipartite}
    Let $G$ be a co-bipartite graph. Then $G$ is $cd$-perfect. 
\end{corollary}
\begin{proof}
Let $G=(A,B,E)$ be a co-bipartite graph. Since the class of co-bipartite graphs is closed under taking induced subgraphs, to prove the corollary, it is enough to show that  $\XCD(G)=\SC(G)$. Note that $\alpha(G)\leq 2$ (as any independent set in $G$ can have only at most one vertex from each of the sets $A$ and $B$). But for each graph $H\in \mathcal{H}$, we have an independent set of size 3 (by Observation~\ref{obs:structureH}). This implies that $G$ is $\mathcal{H}$-free. Since $G^*$ is a bipartite graph (as $G^*$ is a subgraph of complement of $G$), and therefore perfect, we then have $\XCD(G)=k(G^*)=\alpha(G^*)=\SC(G)$, by Theorem~\ref{thm:suff_cdsubclique}. 
\end{proof}

 \noindent\textbf{Note:}
Recall the family of graphs $\mathcal{H}=\{C_6,C_6^1,C_6^2,C_6^3\}$. 
Consider a graph $H\in \mathcal{H}$. It is not difficult to see that $\XCD(H)=3=\SC(H)$, if $H\neq C_6$, and $\XCD(H)=4>2=\SC(H)$, if $H=C_6$. This is why the sufficient condition for $cd$-perfectness given in Corollary~\ref{corr:suffcdperfect} is not a necessary condition. But naturally, if we consider the set $\mathcal{H'}=\{C_6^1,C_6^2,C_6^3\}$, we can obtain a necessary and sufficient condition for an $\mathcal{H'}$-free graph to be $cd$-perfect. This easy consequence (of the above observation, together with Theorem~\ref{thm:suff_cdclique}, Corollary~\ref{corr:suffcdperfect}, and Proposition~\ref{pro:omega_s}) is interesting on its own; because the class of $\mathcal{H'}$-free graphs forms a superclass of triangle-free graphs, $3K_1$-free graphs, etc, and therefore, the following corollary provides a characterization for these classes to be $cd$-perfect.

 \begin{corollary}
Let $G$ be an $\mathcal{H'}$-free graph. Then $G$ is $cd$-perfect if and only if $G$ is $C_6$-free and $k(H^*)=\alpha(H^*)$, for each induced subgraph $H$ of $G$. 
 \end{corollary}

\subsection{\large{\textbf{A unified approach for algorithmic complexity on some special graph classes}}}
\label{sec:poly}
Here, we see some implications of our generalized framework of relating the parameters $\XCD(G)$ and $\SC(G)$ of $G$, respectively to the parameters $k(G^*)$ and $\alpha(G^*)$ of $G^*$. In particular, we use Theorem~\ref{thm:suff_cdclique} and Proposition~\ref{pro:omega_s} as tools for deriving both positive and negative results concerning the algorithmic complexity of the problems \CDC\ and \SCP\ on some special classes of graphs. Note that the classes of graphs for which these results are applied to evaluate the time complexity of the problems \CDC\ and \SCP, may not be limited to those studied here.

 \subsubsection{\textbf{Polynomial-time algorithms:}}
 In this section, we will see a unified approach for obtaining polynomial-time algorithms (mostly with an improvement) for some special classes of graphs, using the framework of $cd$-perfectness and the auxiliary graph $G^*$.
  
\subsubsection*{\textbf{Chordal bipartite graphs}} 
The following theorem is proved in~\cite{ShaluKiruP5Free22}.
\begin{theorem}[\cite{ShaluKiruP5Free22}] \label{thm:shalu}
    Let $G$ be a $P_6$-free chordal bipartite graph. Then $\XCD(G)=\SC(G)$, and the problem \SCP\ can be solved in $O(n^3)$ time for $P_6$-free chordal bipartite graphs.
\end{theorem}

We improve and generalize the result above in Theorem~\ref{thm:chordalbip}.
First, note the following theorem proved in~\cite{DamDomChordBipGr90}.
\begin{theorem}[\cite{DamDomChordBipGr90}] \label{thm:totdom}
    The problem {\sc Total Domination} can be solved in $O(n^2)$ time for chordal bipartite graphs.  
\end{theorem}

We now have the following observations concerning the auxiliary graph of bipartite graphs and chordal bipartite graphs.
 
\begin{observation}\label{obs:bipsquare}
Let $G=(A,B,E)$ be a bipartite graph then $G^*=G^2[A]\cup G^2[B]$. 
\end{observation}
\begin{proof}
    The observation follows from the fact that for any pair of vertices $u$ and $v$ in $G$, we have $uv\in E(G^*)$ (or $d_G(u,v)=2$) if and only if either $u,v\in A$ and $\exists$ $w\in B$ such that $u,v\in N_G(w)$  or alternatively, $u,v\in B$ and $\exists$ $w\in A$ such that $u,v\in N_G(w)$. 
\end{proof}
\begin{observation} \label{obs:chordalbip}
    Let $G=(A,B,E)$ be a chordal bipartite graph. Then, $G^*$ is chordal. 
\end{observation}
\begin{proof}
    Suppose not. Since $G^*=G^2[A]\cup G^2[B]$ (by Observation~\ref{obs:bipsquare}), we have that at least one of the subgraphs $G^2[A]$ or $G^2[B]$ is not chordal. Without loss of generality, we can assume that $G^2[A]$ is not chordal (the other case is symmetric).  Let $C_k=(a_0,a_1\ldots,a_k,a_0)$ be an induced cycle in $G^2[A]$, where $k\geq 3$. Since $G$ is bipartite, and $C_k$ is an induced cycle, this is possible only if for each $i\in \{0,1,\ldots,k\}$, there exists a vertex $b_i$ in $G$ such that $N_G(b_i)\cap V(C_k) = \{a_i,a_{i+1}\}$  (mod $k+1$). This further implies that $(a_0,b_0,a_1,b_1,a_2,\ldots,a_k,b_k,a_0)$ is an induced cycle of length at least $2k$ in $G$. Since $k\geq 3$, this contradicts the fact that $G$ is a chordal bipartite graph. Hence, the observation.
\end{proof}

Now we are ready to prove Theorem~\ref{thm:chordalbip}.\\ %We then have the following theorem.} % \dhanya{due to Theorem~\ref{thm:totdom}}

\noindent\textit{Proof of Theorem~\ref{thm:chordalbip}}
    Clearly, $G$ is $\mathcal{H}$-free. By Observation~\ref{obs:chordalbip}, we have that $G^*$ is chordal, and therefore, perfect. Since any induced subgraph $H$ of $G$ is also chordal bipartite, it follows from Corollary~\ref{corr:suffcdperfect} that $G$ is  $cd$-perfect. Further, as $G$ is triangle-free, by Proposition~\ref{pro:triangle-free}, we have that $\XCD(G)=$\raisebox{2pt}{$\gamma$}$_t(G)$. Since 
 $\XCD(G)=\SC(G)$, the latter statement of the theorem is now immediate from Theorem~\ref{thm:totdom}. \qed

\begin{remark}
    Let $\mathcal{C}=\{G=(A,B,E): G$ is an $\mathcal{H}$-free bipartite graph where $G^2[A]$ and $G^2[B]$ are perfect\}. As in the proof of the above theorem, by Observation~\ref{obs:bipsquare} and Corollary~\ref{corr:suffcdperfect}, we can observe that {\sc Total Domination}, \CDC, and \SCP\  are all equivalent problems on $\mathcal{C}$. Moreover, since the problem \CC\ is polynomial-time 
 solvable for perfect graphs~\cite{GrotSchrEllipsoid81}, by Theorem~\ref{thm:suff_cdclique}, we then have that all the problems {\sc Total Domination}, \CDC, and \SCP\ can be solved in polynomial time for graphs in $\mathcal{C}$.
\end{remark}

\subsubsection*{\textbf{Proper interval graphs and $3K_1$-free graphs}}

%\label{subsec:Proper-interval graphs}
It is known that for the classes of proper interval graphs and $3K_1$-free graphs, the \CDC\ problem can be solved in $O(n^3)$ time~\cite{ShaluVSCDComplex20} (as they are sub-classes of claw-free graphs). We show that the \CDC\ problem can be solved in $O(n^{2.5})$ time for the above classes of graphs. 

\begin{observation} \label{obs:propertriangle}
Let $G$ be a proper interval graph or a $3K_1$-free graph. Then, $G^*$ is triangle-free.
\end{observation}

\begin{proof}
Let $G$ be a proper interval graph or a $3K_1$-free graph. For the sake of contradiction, assume that $K=\{a,b,c\}\subseteq V(G)=V(G^*)$ is a triangle in $G^*$. First, we claim that $G$ is $\mathcal{H}$-free. Note that for each graph $H\in \mathcal{H}$, we have an independent set of size 3, and therefore, a $3K_1$, and an asteroidal-triple (by Observation~\ref{obs:structureH}). Since $G$ is either a proper interval graph (which is AT-free) or a $3K_1$-free graph, we can therefore conclude that $G$ is an $\mathcal{H}$-free graph. 
%Now, as $K$ induces a triangle in $G^*$, we have that $K$ is an independent set in $G$ (by the definition of $G^*$). 
Since $K$ is a clique in $G^*$, we have $K$ is an independent set in $G$ (by the definition of $G^*$)  and as $G$ is $\mathcal{H}$-free, by Lemma~\ref{lem:neighborhood}, there exists a vertex $v_c\in V(G)$ such that $K\subseteq N_G(v_c)$. This implies that the graph induced by the set $K\cup \{v_c\}$ is a claw in $G$, a contradiction to the fact that $G$ is a proper interval graph or a $3K_1$-free graph (as they are both sub-classes of claw-free graphs).   
\end{proof}

Before evaluating the complexity of the \CDC\ problem for proper interval graphs and $3K_1$-free graphs, we note a few observations and a lemma. Note that for any graph $G$, we have $V(G^*)=V(G)=V(G^2)$, and $E(G^*)=E(G^2)\setminus E(G)$. Therefore, if $G^2$ can be computed in $O(n^2)$ time, then $G^*$ can also be computed in $O(n^2)$ time (as the subtraction of the adjacency matrices of $G^2$ and $G$ takes only $O(n^2)$ time). The first observation below is due to the result that the squares of proper interval graphs can be computed in $O(n^2)$ time~\cite{SatPMadPLinAlgoSqGr16}. 
\begin{observation}\label{obs:proper}
    For a proper interval graph $G$, the auxiliary graph $G^*$ can be computed in $O(n^2)$ time. 
\end{observation}
We now have some additional observations.
\begin{observation}\label{obs:cobipartite}
    For a co-bipartite graph $G$, the auxiliary graph $G^*$ can be computed in $O(n^2)$ time.
\end{observation}
\begin{proof}
    Let $G$ be a co-bipartite graph. Note that for any pair of vertices $u,v\in V(G)$, we have $uv\in E(G^2)$ if and only if $u\in A$, $v\in B$, and either $N_B(u)\neq \emptyset$ or $N_A(v)\neq \emptyset$ or both. This implies that $G^2$ can be constructed in $O(n^2)$ time, and so does $G^*$. \end{proof}
In the following lemma, we observe a structural property of $3K_1$-free graphs, which is crucial in proving Observation~\ref{obs:alpha} and is of independent interest.
\begin{lemma}\label{lem:nocobipatite}
    Let $G$ be a $3K_1$-free graph. If $G$ is not co-bipartite, then $G^2$ is a clique.
\end{lemma}
\begin{proof}
    
    Let $u$ and $v$ be any two non-adjacent vertices in $G$. To prove the lemma, it is enough to show that $uv\in E(G^2).$ Let $d_G(u)$ and $d_G(v)$ denote the degrees of the vertices $u$ and $v$ in $G$, respectively.  We first claim that $d_G(u)+d_G(v)\geq n-2$. Suppose not. Let $d_G(u)+d_G(v) < n-2$. This implies that $\mid N_G(u)\cup N_G(v)\cup\{u,v\}\mid <n$. This further implies that there exists a vertex $w\in V(G)$ such that $w$ is non-adjacent to both $u$ and $v$. Therefore, $\{u,v,w\}$ forms an independent set in $G$, which contradicts the fact that $G$ is $3K_1$-free. Therefore, we can assume that $d_G(u)+d_G(v)\geq n-2$.

    Suppose that $d_G(u)+d_G(v)> n-2$. Then, as $\mid V(G)\setminus \{u,v\}\mid \leq n-2$, it should be the case that $N_G(u)\cap N_G(v)\neq \emptyset$. i.e. there exists a vertex $w\in V(G)$ such that $uw,wv\in E(G)$.  This implies that $uv\in E(G^2)$, and we are done. Now consider the case that $d_G(u)+d_G(v)= n-2$. Suppose that $N_G(u)\cap N_G(v)=\emptyset$. Then, as $\mid N_G(u)\cup N_G(v)\cup \{u,v\}\mid =n$, we have that $V(G)$ is a disjoint union of two sets $A=N_G(u)\cup \{u\}$ and $B=N_G(v)\cup \{v\}$. Since $G$ is not a co-bipartite graph, we can assume that at least one of the sets $N_G(u)$ or $N_G(v)$ contains at least two vertices in it.  Suppose that $\mid N_G(u)\mid \geq 2$ (respectively $\mid N_G(v)\mid \geq 2$). Now, if there exist vertices $x,y\in N_G(u)$ (respectively $N_G(v)$) such that $xy\notin E(G)$, then, as $x,y\notin N_G(v)$ (respectively $N_G(u)$), we have that $\{x,y,v\}$ (respectively $\{x,y,u\}$) forms an independent set in $G$. As this contradicts the fact that  $G$ is $3K_1$-free, we have that both the sets $N_G(u)$ and $N_G(v)$ form a clique. But then, this contradicts the fact that $G$ is not a co-bipartite graph (as in this case, $A$ and $B$ clearly form a two-clique partition of $V(G)$). Thus we can conclude that $N_G(u)\cap N_G(v)\neq \emptyset$, and therefore, $uv\in E(G^2)$. This completes the proof of lemma.
    \end{proof}
    \begin{observation} \label{obs:alpha}
        For a $3K_1$-free graph $G$, the auxiliary graph $G^*$ can be computed in $O(n^2)$ time.
    \end{observation}
  \vspace{-0.5cm}
    \begin{proof}
        Let $G$ be a $3K_1$-free graph. We can check whether $G$ is co-bipartite or not in $O(n^2)$ time (as it is enough to check whether $\overline{G}$ is bipartite or not). If $G$ is co-bipartite, then by Observation~\ref{obs:cobipartite}, we have that $G^*$ can be computed in $O(n^2)$ time. If $G$ is not co-bipartite, by  Lemma~\ref{lem:nocobipatite}, we have that $G^2$ is a clique, and therefore, the auxiliary graph $G^*$, can be computed in $O(n^2)$ time.
    \end{proof}
 Now we are ready to prove Theorem~\ref{thm:properalpha}.
    \vspace{0.5cm}
    
\noindent\textit{Proof  Theorem~\ref{thm:properalpha}}
Let $G$ be a proper interval graph or a $3K_1$-free graph. As noted in the proof of Observation~\ref{obs:propertriangle}, we have that $G$ is $\mathcal{H}$-free. Therefore, by Theorem~\ref{thm:suff_cdclique}, we have that $\XCD(G)=k(G^*)$, where $G^*$ is the corresponding auxiliary graph. By Observation~\ref{obs:propertriangle}, we have that $G^*$ is triangle-free. This implies that each clique in any clique cover of $G^*$ is either an edge or a single vertex. This further implies that $k(G^*)=n-\mid M\mid$, where $M$ denotes the size of the maximum cardinality matching in $G^*$. Note that if $G$ is a proper interval graph, then the auxiliary graph $G^*$ can be constructed in $O(n^2)$ time by Observation~\ref{obs:proper}, and if $G$ is a  $3K_1$-free graph, then the auxiliary graph $G^*$ can be constructed in $O(n^2)$ time by Observation~\ref{obs:alpha}. Also, the size of the maximum 
 cardinality matching in $G^*$ can be computed in $O(\sqrt{n}m')$ time~\cite{HopKarAlgoMaxMat73}, where $m'=\mid E(G^*)\mid \leq n^2$. Since the overall complexity is $O(n^{2.5})$, this proves the theorem. \qed

\begin{remark}
    As noted before, the time complexity of \CDC\ for the class of proper interval graphs and $3K_1$-free graphs provided in Theorem~\ref{thm:properalpha} is an improvement over the existing algorithms~\cite{ShaluVSCDComplex20} for the problem in the same graph classes. Clearly, the class of $3K_1$-free graphs is not $cd$-perfect, as \SCP\ is \NPC\ for $3K_1$-free graphs~\cite{ShaluSandhyaLowBound17}. %In Section~\ref{sec:interval}, we give a polynomial-time algorithm for \SCP\ problem for the class of interval graphs. 
    Further, we note that the parameters $\XCD$ and $\SC$ are not necessarily equal for proper interval graphs too (see an example in Figure~\ref{fig:proper_interval}). Therefore, the class of proper interval graphs and, in general, interval graphs are not $cd$-perfect. Now, regarding the complexity of \SCP\ problem on the class of proper interval graphs, we prove in Section~\ref{sec:interval} that \SCP\ is polynomial-time solvable even for its superclass, namely, interval graphs. 
\end{remark}

\begin{figure}[ht]
  \centering
    \centering
    \begin{tikzpicture}[myv/.style={circle, draw, inner sep=1.5pt,line width=0.3mm}]
  \node (z) at (0,0) {};

  \node[myv] (a) at (-1,0) {};
  \node[myv] (b) at (0,0) {};
  \node[myv] (c) at (1,0) {};
  \node[myv] (d) at (2,0) {};
  \node[myv] (e) at (3,0) {};
  \node[myv] (f) at (0.5,0.5) {};
  \node[myv] (g) at (1.5,0.5) {};
 
  \draw[line width=0.3mm] (a) -- (b);
  \draw[line width=0.3mm] (b) -- (c);
  \draw[line width=0.3mm] (c) -- (d);
  \draw[line width=0.3mm] (d) -- (e);
  \draw[line width=0.3mm] (b) -- (f);
  \draw[line width=0.3mm] (c) -- (f);
  \draw[line width=0.3mm] (c) -- (g);
  \draw[line width=0.3mm] (d) -- (g);
\end{tikzpicture}
    \caption{An example of a proper interval graph $G$,  with $\XCD(G)=4>3=\omega_s(G)$.} 
    \label{fig:proper_interval}
  \end{figure}

\subsubsection{\textbf{Hardness implications:}}
In contrast to the previous section, here, we see an application of Proposition~\ref{pro:omega_s} and Theorem~\ref{thm:suff_cdclique} in deriving hardness results for \CDC\ and \SCP. 

\subsubsection*{\textbf{$C_6$-free bipartite graphs}}

%\label{subsec:c6-free}
Here, we prove the hardness of the problems \CDC\ and \SCP\ for $C_6$-free bipartite graphs by proposing a polynomial-time reduction from the problems \CC\  and \ISP\ for diamond-free graphs. We begin by stating one of the main results in this section. 
%The proof of this result is the main content of this subsection.

Note that both {\sc Clique Cover} and {\sc Independent Set} are known to be NP-hard for diamond-free graphs~\cite{KralKraComplexColGr01,PoljakStabSet74}. Recall that, a \textit{diamond} is a graph obtained by deleting an edge from $K_4$ (complete graph on 4 vertices). First, we note the following observation for diamond-free graphs~\cite{ChiarStrongCli21}.

\begin{observation}[\cite{ChiarStrongCli21}]\label{obs:diamondfree}
Any diamond-free graph with $m$ edges can have at most $m$ maximal cliques.
\end{observation}
Let $G$ be a diamond-free graph and let $\mathcal{C}=\{C_1,C_2,\ldots,C_l\}$ denote the collection of all maximal cliques in $G$. By Observation~\ref{obs:diamondfree}, we have that $l\leq~ \mid~E(G)\mid$.  From $G$, we now construct a bipartite graph $B_G=(A,B,E)$ in polynomial-time as follows:
\begin{construction}  \label{cons:c6-freebipartite}
     Define $A=V(G)\cup\{u\}$ and $B=\mathcal{C}=\{C_1,C_2,\ldots,C_l\}$ (each vertex in the partite set $B$ represents a maximal clique in $G$). For a pair of vertices, $a\in A\setminus \{u\}$ and $C_j\in B$ (where $j\in \{1,2,\ldots,l\}$), we make the vertices $a$ and $C_j$ adjacent in $B_G$ if and only if $a\in C_j\subseteq V(G)$. In addition, we make the vertex $u\in A$ adjacent to all the vertices in $B$. i.e. $E(B_G)=\{aC_j:a\in A\setminus \{u\}, C_j\in \mathcal{C}=B$ with $a\in C_j\subseteq V(G)\}\cup \{uC_j:C_j\in \mathcal{C}=B\}$. Clearly, $B_G$ is a bipartite graph.
\end{construction}

% Example Shirish 

 \begin{figure}[!htbp]
 \centering
    \begin{subfigure}[b]{0.45\textwidth}
          \centering
         \usetikzlibrary{arrows}
\usetikzlibrary{decorations.markings}
\usetikzlibrary{shapes.geometric}
\usetikzlibrary{positioning}
%\usetikzlibrary{decorations.text}
%\usetikzlibrary{decorations.pathmorphing}

\resizebox{\textwidth}{!}{\begin{tikzpicture} [myv1/.style={circle, draw, inner sep=3.5pt, line width = 1 mm, fill=black},myv4/.style={rectangle, draw,dotted,inner sep=0pt,line width = 0.2 mm},mvy5/.style={ellipse, draw, dotted, line width = 0.2mm},mydummy/.style={circle, inner sep=0pt, fill=white}];

\tikzset{
  big arrow/.style={
    decoration={markings,mark=at position 1 with {\arrow[scale=4,#1]{>}}},
    postaction={decorate},
    shorten >=0.4pt},
  big arrow/.default=black}

\node[myv1] [label = left: \text{\huge {$v_1$}}] (v1) at (1, 1) {};
\node[mydummy] [label={}] (v11a) at (1, 1.5) {};
\node[mydummy] [label={}] (v12a) at (1, 2) {};

\node[myv1] [label = right: \text{\huge { $v_2$}}] (v2) at (5, 1) {};
\node[mydummy] [label={}] (v21a) at (5, 1.5) {};
\node[mydummy] [label={}] (v22a) at (5, 2) {};

\node[myv1] [label = left: \text{\huge {$v_3$}}] (v3) at (-1, -1) {};
%\node[mydummy] (v31b) at (-1, -1.5);
\node[mydummy] [label={}] (v32b) at (-1, -2) {};
%\node[mydummy] (v31l) at (-1.5, -1 );
\node[mydummy] [label={}](v32l) at (-2, -1 ) {};

\node[myv1] [label = above: \text{\huge {$v_4$}}] (v4) at (3, -1) {};
\node[mydummy] [label={}] (v41a) at (3, -0.5) {};
\node[mydummy] [label={}] (v41b) at (3, -1.5) {};
\node[mydummy] [label={}] (v42b) at (3, -2) {};
\node[mydummy] [label={}] (v41l) at (2.5, -1) {};
\node[mydummy] [label={}] (v41r) at (3.5, -1) {};

\node[myv1] [label = right: \text{\huge {$v_5$}}] (v5) at (7, -1) {};
%\node[mydummy] (v51b) at (7, -1.5);
\node[mydummy] [label={}] (v52b) at (7, -2) {};
%\node[mydummy] (v51r) at (7.5,-1);
\node[mydummy] [label={}] (v52r) at (8,-1) {};

\node[myv1] [label = left: \text{\huge {$v_6$}}] (v6) at (-3, -3) {};
%\node[mydummy] (v61l) at (-3.5,-3);
\node[mydummy] [label={}] (v62l) at (-4,-3) {};

\node[myv1] [label = above: \text{\huge {$v_7$}}] (v7) at (1, -3) {};
\node[mydummy] [label={}] (v71l) at (0.5, -3) {};
%\node[mydummy] [label=] (v72l) at (0, -3) {};

\node[myv1] [label = above: \text{\huge {$v_8$}}] (v8) at (5, -3) {};
%\node[mydummy] (v81r) at (5.5,-3);
\node[mydummy] [label={}] (v82r) at (6,-3) {};

\node[myv1] [label = right: \text{\huge {$v_9$}}] (v9) at (9, -3) {};
%\node[mydummy] (v91r) at (9.5, -3);
\node[mydummy] [label={}] (v92r) at (10, -3) {};

\draw (v1) -- (v3) {};
\draw (v1) -- (v4) {};
\draw (v2) -- (v4) {};
\draw (v2) -- (v5) {}; 
\draw (v3) -- (v4) {};
\draw (v3) -- (v6) {};
\draw (v4) -- (v5) {};
\draw (v4) -- (v7) {};
\draw (v4) -- (v8) {};    
\draw (v5) -- (v9) {};
\draw (v7) -- (v8) {};

\node[myv4][fit=(v11a) (v1) (v3) (v32l) (v4) (v41r),  inner xsep=1.5ex, inner ysep=1.5ex, label=above: \text{\huge {$C_1$}}] {}; 

\node[myv4][fit=(v2) (v4) (v5) (v52r) (v41b),  inner xsep=1.5ex, inner ysep=1.5ex, label=above: \text{\huge { $C_2$}}] {}; 

\node[myv4][fit=(v3) (v6) (v62l) ,  inner xsep=1.5ex, inner ysep=1.5ex, label=left: \text{\huge {$C_3$}}] {}; 

\node[myv4][fit=(v4) (v7) (v8) (v41a),  inner xsep=1.5ex, inner ysep=1.5ex, label=below: \text{\huge {$C_4$}}] {}; 

\node[myv4][fit=(v5) (v9) (v92r),  inner xsep=1.5ex, inner ysep=1.5ex, label=right: \text{\huge {$C_5$}}] {}; 

\end {tikzpicture}} 
          \caption{A diamond-free graph $G$}
          \label{fig:diamond-free}
     \end{subfigure}
      \begin{subfigure}[b]{0.45\textwidth}
          \centering
         \usetikzlibrary{arrows}
\usetikzlibrary{decorations.markings}
\usetikzlibrary{shapes.geometric}
\usetikzlibrary{positioning}
%\usetikzlibrary{decorations.text}
%\usetikzlibrary{decorations.pathmorphing}

\resizebox{0.40\textwidth}{!}{\begin{tikzpicture} [myv1/.style={circle, draw, inner sep=3.5pt, line width = 1 mm, fill=black},myv4/.style={rectangle, draw,dotted,inner sep=0pt,line width = 0.2 mm},mvy5/.style={ellipse, draw, dotted, line width = 0.2mm,inner sep=5pt},mydummy/.style={circle, inner sep=0pt, fill=white}];

\tikzset{
  big arrow/.style={
    decoration={markings,mark=at position 1 with {\arrow[scale=4,#1]{>}}},
    postaction={decorate},
    shorten >=0.4pt},
  big arrow/.default=black}

\node[myv1] [label = left: \text{\huge{ $v_1$}}] (v10) at (17, 6) {};
\node[myv1] [label = left: \text{ \huge{$v_2$}}]  (v11) at (17, 4.75) {};
\node[myv1] [label = left: \text{\huge{ $v_3$}}] (v12) at (17, 3.25) {};
\node[myv1] [label = left: \text{\huge{ $v_4$}}] (v13) at (17, 2) {};
\node[myv1] [label = left: \text{\huge{$v_5$}}] (v14) at (17, 0.75) {};
\node[myv1] [label = left: \text{\huge{ $v_6$}}] (v15) at (17, -0.5) {};
\node[myv1] [label = left: \text{\huge{ $v_7$}}] (v16) at (17, -1.75) {};
\node[myv1] [label = left: \text{\huge{ $v_8$}}] (v17) at (17, -3) {};
\node[myv1] [label = left: \text{\huge{$v_9$}}] (v18) at (17, -4.25) {};
\node[myv1] [label = left: \text{\huge{ $u$}}] (u) at (17, -6) {};

\node[myv1] [label = right: \text{\huge{ $C_1$}}] (k6) at (23, 4) {};
\node[myv1] [label = right: \text{\huge{ $C_2$}}] (k7) at (23, 2) {};
\node[myv1] [label = right: \text{\huge{ $C_3$}}] (k8) at (23, 0) {};
\node[myv1] [label = right: \text{\huge{ $C_4$}}] (k9) at (23, -2) {};
\node[myv1] [label = right: \text{\huge{ $C_5$}}] (k10) at (23, -4) {};

\node [mvy5] [label= {}, fit=(v10)(v11)(v12)(v13)(v14)(v15)(v16)(v17)(v18)(u), inner ysep=1ex, inner xsep=8ex] {};

\node [mvy5] [label={}, fit=(k6)(k7)(k8)(k9)(k10), inner ysep=4ex, inner xsep=8ex] {};

\draw (v10) -- (k6) {};

\draw (v11) -- (k7) {};

\draw (v12) -- (k6) {};
\draw (v12) -- (k8) {};

\draw (v13) -- (k6) {};
\draw (v13) -- (k7) {};
\draw (v13) -- (k9) {};

\draw (v14) -- (k7) {};
\draw (v14) -- (k10) {};

\draw (v15) -- (k8) {};

\draw (v16) -- (k9) {};
\draw (v17) -- (k9) {};

\draw (v18) -- (k10) {};

\draw (u) -- (k6) {};
\draw (u) -- (k7) {};
\draw (u) -- (k8) {};
\draw (u) -- (k9) {};
\draw (u) -- (k10) {};

\end {tikzpicture}}  
          \caption{Corresponding bipartite graph $B_G$}
          \label{fig:c6-free bipartite}
     \end{subfigure}
    
     \caption{ An illustration of  Construction~\ref{cons:c6-freebipartite}.}
    \label{fig:diamond-free to bipartite}
\end{figure}

  An example of  Construction~\ref{cons:c6-freebipartite} is shown in Figure~\ref{fig:diamond-free to bipartite}. We then have the following lemmas, which are crucial for the reduction. 

\begin{lemma}\label{lem:isomprphic}
Let $G$ be a diamond-free graph and $B_G=(A,B,E)$ the corresponding bipartite graph (as in Construction~\ref{cons:c6-freebipartite}). Then, $B_G^2[A\setminus \{u\}]\cong G$ and $B_G^2[B]$ is a clique.
\end{lemma}
\begin{proof}
Let $H=B_G^2[A\setminus \{u\}]$. Clearly, by the definition of $B_G$, every vertex in $V(H)$ corresponds to a vertex in $V(G)$. Let $x,y\in V(H)$. Then, as $B_G$ is a bipartite graph, and $x,y\in V(H)\subseteq A\setminus \{u\}$, we have $xy\in E(H)\iff \exists C_j\in B$ for some $j\in \{1,2,\ldots,l\}$ such that $x,y\in N_{B_G}(C_j)$ (i.e. both the vertices $x$ and $y$ belong to a same maximal clique $C_j$ in $G$) $\iff xy\in E(G)$. This proves that $B_G^2[A\setminus \{u\}]\cong G$. Further, it is easy to see that $B_G^2[B]$ is a clique, since $B\subseteq N_{B_G}(u)$.
\end{proof}

\begin{lemma}\label{lem:diamondC_6}
    Let $G$ be a diamond-free graph and $B_G=(A,B,E)$ be the corresponding bipartite graph obtained from Construction~\ref{cons:c6-freebipartite}. Then, $B_G$ is a $C_6$-free bipartite graph.
\end{lemma}
\begin{proof}
    Suppose not. Let $S=(a_1,C_1,a_2,C_2,a_3,C_3,a_1)$ be an induced $C_6$ in $B_G$, where $\{a_1,a_2,a_3\}\subseteq A$ and $\{C_1,C_2,C_3\}\subseteq B$.  Clearly, $\{a_1,a_2,a_3\}$ induces a triangle in $B_G^2[A]$. Since each of the vertex $a_i$ has a non-neighbor in $V(S)\cap B$, we have $u\neq a_i$ for any $i\in \{1,2,3\}$. This implies that $\{a_1,a_2,a_3\}$ induces a triangle in $B_G^2[A\setminus \{u\}]$, and thus a triangle in $G$ by Lemma~\ref{lem:isomprphic}. Now, consider the vertex $C_1$ in the cycle $S$. We have $N_{B_G}(C_1)\cap V(S) =\{a_1,a_2\}$. Clearly, $\{a_1,a_2\}$ is not a maximal clique in $G$, since $\{a_1,a_2,a_3\}$ induces a triangle in $G$. As $a_3\notin N_{B_G}(C_1)$, we have by the definition of $B_G$ that $a_3\notin C_1$ in $G$. Since $C_1$ is a maximal clique in $G$ containing the set $\{a_1,a_2\}$, but not the vertex $a_3$, there exists a vertex $a_4\in V(G)\subseteq A$ such that $a_4\in C_1$ in $G$ and $a_3a_4\notin E(G)$. This implies that $\{a_1,a_2,a_3,a_4\}$ induces a diamond in $G$, a contradiction to the fact that $G$ is diamond-free.
\end{proof}
We are now ready to prove Theorem~\ref{thm:C_6free hard}.\\

\noindent\textit{Proof of Theorem~\ref{thm:C_6free hard}}.
    Let $G$ be a diamond-free graph, and let $B_G$ be the corresponding bipartite graph obtained from Construction~\ref{cons:c6-freebipartite}. Since $B_G$ is $C_6$-free (by Lemma~\ref{lem:diamondC_6}) and triangle-free, we can therefore conclude that $B_G$ is $\mathcal{H}$-free. Therefore, by Theorem~\ref{thm:suff_cdclique}, we have that $\XCD(B_G)=k(B_G^*)$. Since $B_G^*= B_G^2[A]\uplus B_G^2[B]$, we then have $\XCD(B_G)=k(B_G^*) = k(B_G^2[A])+k(B_G^2[B])$ = $k(B_G^2[A\setminus \{u\}])+k(B_G^2[B])$ (again, the last equality is due to the fact that, in the graph $B_G^2[A]$, the vertex $u$ is adjacent to very vertex in $A\setminus \{u\}$). This implies that $\XCD(B_G)=k(B_G^*) = k(G)+1$ (by Lemma~\ref{lem:isomprphic}). Recall that \CC\ is NP-hard for diamond-free graphs. This implies that \CDC\ is NP-hard for $C_6$-free bipartite graphs.

    By Proposition~\ref{pro:omega_s} applied to the $C_6$-bipartite graph $B_G$, we have that $\SC(B_G)=\alpha(B_G^*)$, where $B_G^*=B_G^2- B_G$ = $B_G^2[A]\uplus B_G^2[B]$. Therefore,  we have $\SC(B_G)=\alpha(B_G^*) = \alpha(B_G^2[A])+\alpha(B_G^2[B])$ = $\alpha(B_G^2[A\setminus \{u\}])+\alpha(B_G^2[B])$ (the last equality is due to the fact that, in the graph $B_G^2[A]$, the vertex $u$ is adjacent to every vertex in $A\setminus \{u\}$). This implies that $\SC(B_G)=\alpha(B_G^*) = \alpha(G)+1$ (by Lemma~\ref{lem:isomprphic}). %Note that $B_G$ is a $C_6$-free bipartite graph (by Lemma~\ref{lem:diamondC_6}).
    Recall that \ISP\ is NP-hard for diamond-free graphs. This implies that \SCP\ is NP-hard for $C_6$-free bipartite graphs. 
    Hence, the theorem.\qed

\subsection{\large{\textbf{Some necessary conditions for $cd$-perfectness}}}

Let $C_n$ denote an induced cycle of length $n$ and $\bar{C_n}$ its complement. By the Strong Perfect Graph theorem, we have that $C_{n}$ or $\bar{C_{n}}$ is perfect if and only if $n=2k$ for some $k\geq 2$. In an upcoming theorem, we give an \textit{almost similar} necessary condition for $cd$-perfect graphs. The following observation for the cycles is noted in~\cite{ShaluSandhyaLowBound17}.
\begin{observation}[\cite{ShaluSandhyaLowBound17}]\label{obs:holes}
 For $n\geq 4$, we have  $\XCD (C_n)=\SC(C_n)$ if and only if $n=4k$ for some integer $k\geq 1$.
\end{observation}
We now have a similar observation for the complements of cycles.
%As in perfect graphs, we note the following necessary condition for $cd$-perfect graphs.
\begin{observation}\label{obs:antiholes}
For $n\geq 5$, we have  $\XCD (\bar{C_n})=\SC(\bar{C_n})$ if and only if $n=2k$ for some integer $k\geq 2$.
\end{observation}
\begin{proof}
First, note that for each $n\geq 5$,  $\bar{C_n}$ is $\mathcal{H}$-free (since $\alpha(\bar{C_n})\leq 2$, for each $n\geq 5$, but $\alpha(H)=3$ for each $H\in \mathcal{H}$). Also, for each $n\geq 5$, it is easy to see that $\bar{C_n}^*=C_n$ (since $\bar{C_n}^2$ is a clique on $n$ vertices). Therefore, by Theorem~\ref{thm:suff_cdclique} and Proposition~\ref{pro:omega_s}, we have, $\XCD(\bar{C_n})=k(\bar{C_n}^*)=k(C_n)=\lceil \frac{n}{2}\rceil$ and $\SC(\bar{C_n})=\alpha(\bar{C_n}^*)=\alpha(C_n)=\lfloor \frac{n}{2}\rfloor$. Therefore, we can conclude that $\XCD (\bar{C_n})=\SC(\bar{C_n})$ if and only if $n=2k$ for some integer $k\geq 2$.
\end{proof}
As noted earlier, in the following theorem, we summarize some necessary conditions for a graph $G$ to be $cd$-perfect. 
%Observe that these conditions are \textit{almost} consistent with the necessary conditions for a graph to be perfect. 
The proof of the theorem is immediate from Observations~\ref{obs:holes} and~\ref{obs:antiholes}.

\begin{theorem}\label{thm:necessary}
If a graph $G$ is $cd$-perfect, then $G$ is $C_n$-free for each $n\geq 4$ with $n\neq 4k$, and $\bar{C_n}$-free for each $n\geq 5$ with $n\neq 2k$, where $k$ is a positive integer.
\end{theorem}

\section{Separated-cluster problem in Interval graphs}
\label{sec:interval}In this section, we settle an open problem proposed by Shalu et al.~\cite{ShaluSandhyaLowBound17}, by proving that \SCP\ is polynomial-time solvable for the class of interval graphs (Theorem~\ref{thm:sepiterval}). We achieve this by showing a polynomial-time reduction of \SCP\ problem on interval graphs to the maximum weighted independent set problem on cocomparability graphs (which can be solved in time linear to the size of the input graph~\cite{KohLalMaxWtIndSetAlg16}). An undirected graph is called \textit{comparability graph}, if it is transitively orientable, i.e., its edges can be directed such that if $a\rightarrow b$ and $b\rightarrow c$ are directed edges, then $a\rightarrow c$ is a directed edge. \textit{Cocomparability graphs} are complements of comparability graphs. Throughout the section, we use an alternative definition of \textit{separated-cluster}. i.e. for a graph $G$, we say that a family of disjoint cliques in $G$, forms a \textit{separated-cluster}  $\mathcal{S}=\{S_1,S_2,\ldots,S_k\}$, if any pair of cliques in $\mathcal{S}$ are mutually non-adjacent and no two vertices belonging to distinct cliques in $\mathcal{S}$ have a common neighbor in $G$. Formally, for $i,j\in \{1,2,\ldots,k\}$ with $i\neq j$, if $x\in S_i$ and $y\in S_j$, then $xy\notin E(G)$, and there does not exist a vertex $z\in V(G)\setminus (S_i\cup S_j)$ such that $x,y\in N_G(z)$.  \\

\noindent\textbf{Sketch of the reduction: }
Let $G$ be an interval graph with an interval representation, $\{I_v\}_{v\in V(G)}$, where $n=~ \mid V(G) \mid $. 
Without loss of generality, we can assume that \textit{all the end-points of the intervals in $\{I_v\}_{v\in V(G)}$ are distinct}. 
For a vertex $v\in V(G)$, for the sake of convenience, we denote by $l(v)$ and $r(v)$, the \textit{left} and \textit{right end-points} of the interval $I_v$, respectively. Our basic idea for the reduction is as follows. Recall that in \SCP\ problem, we need to find a collection of cliques in $G$ satisfying certain conditions. First, note that the cliques belonging to a separated-cluster are not necessarily maximal in $G$. Therefore, our primary goal is to find a polynomial (in $n$) sized collection of cliques in $G$ that contains all the cliques that may possibly belong to {some} maximum cardinality separated-cluster in $G$. To achieve this, we heavily use the interval representation of $G$, and extend the collection of maximal cliques in $G$ to a \textit{``special"} collection of cliques denoted as $\hat{\mathcal{C}}$. Each member of $\hat{\mathcal{C}}$ is obtained by \textit{``pruning"} a maximal clique in $G$ with respect to a fixed interval representation of $G$. Then, from the input graph $G$, we construct a weighted conflict graph $G_c$, having the collection of cliques $\hat{\mathcal{C}}$ as the vertex set, and the vertices are assigned a weight same as the cardinalities of the corresponding sets (when they are considered as cliques in $G$). The graph $G_c$ is referred to as a conflict graph because, each vertex in $G_c$ represents a clique in $G$, and the edges between two vertices are added to $G_c$ if the corresponding cliques can not appear together in any separated-cluster.  We then prove that the weighted conflict graph $G_c$ is a \textit{cocomparability} graph and \textit{a maximum cardinality separated-cluster in the interval graph $G$ is equivalent to the maximum cardinality independent set in $G_c$}. 

% Example
\begin{figure}
  \centering
    \centering
 %  \vspace{-0.5cm}
    \usetikzlibrary{arrows}
\usetikzlibrary{decorations.markings}
\usetikzlibrary{shapes.geometric}
\usetikzlibrary{positioning}
\usetikzlibrary{decorations.text}
\usetikzlibrary{decorations.pathmorphing}

\resizebox{0.80\textwidth}{!}{\begin{tikzpicture}[myv1/.style={circle, draw, inner sep=3.5pt, line width = 1 mm, fill=black},myv4/.style={rectangle, draw,dotted,inner sep=0pt, line width = 0.2mm},mvy5/.style={ellipse, draw, line width = 1mm},mydummy/.style={circle, inner sep=0pt, fill=white}]

\tikzset{
  big arrow/.style={
    decoration={markings,mark=at position 1 with {\arrow[scale=4,#1]{>}}},
    postaction={decorate},
    shorten >=0.4pt},
  big arrow/.default=black}

 % \node [label=above: $\Large{r}$](z)  at (0,4) {};
\node[myv1] [label = left: \text{\Large $v_1$}] (v1) at (0, 0) {};
\node[mydummy] [label=] (v11l) at (-1, 0) {};
\node[mydummy] [label=] (v11b) at (0, -0.5) {};

\node[myv1] [label = right: \text{\Large $v_2$}] (v2) at (2, 2) {};
\node[mydummy] [label=] (v21a) at (2, 2.5) {};
\node[mydummy] [label=] (v21l) at (1.5, 2) {};

\node[myv1] [label = below: \text{\Large $v_3$}] (v3) at (4, 0) {};
\node[mydummy] [label=] (v31b) at (4, -0.5) {};
\node[mydummy] [label=] (v32b) at (4, -1) {};
\node[mydummy] [label=] (v33b) at (4, -1.5) {};
\node[mydummy] [label=] (v31a) at (4, 0.5) {};
\node[mydummy] [label=] (v31l) at (3.5, 0) {};
\node[mydummy] [label=] (v31r) at (4.5, 0) {};

\node[myv1] [label = right: \text{\Large $v_4$}] (v4) at (6, 2) {};
%\node[mydummy] [label=] (v41a) at (6, 1.5);
\node[mydummy] [label=] (v42a) at (6, 0.5) {};

\node[myv1] [label = below: \text{\Large $v_5$}] (v5) at (8, 0) {};
\node[mydummy] [label=] (v51b) at (8, -0.5) {};

\node[myv1] [label = below: \text{\Large $v_6$}] (v6) at (10, 0) {};
\node[mydummy] [label=] (v62b) at (10, -1) {};

\node[myv1] [label = above: \text{\Large $v_7$}] (v7) at (10, 2) {};
\node[mydummy] [label=] (v71a) at (10, 2.5) {};
\node[mydummy] [label=] (v72a) at (10, 3) {};

\node[myv1] [label = above: \text{\Large $v_8$}] (v8) at (12, 2) {};

\node[myv1] [label = below: \text{\Large $v_9$}] (v9) at (12, 0) {};
\node[mydummy] [label=] (v92b) at (12, -0.5) {};

\node[myv1] [label = below: \text{\Large $v_{10}$}] (v10) at (14, 0) {};
\node[mydummy] [label=] (v101a) at (14, 0.5) {};
\node[mydummy] [label=] (v102b) at (14, -1) {};

\draw [label=](v1) -- (v3);
\draw [label=](v2) -- (v3);
\draw [label=](v3) -- (v4);
\draw [label=](v3) -- (v5); 
\draw [label=](v4) -- (v5);
\draw [label=](v5) -- (v6);
\draw [label=](v6) -- (v7);
\draw [label=](v6) -- (v8);
\draw [label=](v6) -- (v9);
\draw [label=](v7) -- (v8);
\draw [label=](v7) -- (v9);
\draw [label=](v8) -- (v9);
\draw [label=](v9) -- (v10);

\node[myv4][fit= (v11l) (v1) (v3) (v31b) (v31r),  inner xsep=1.5ex, inner ysep=1.5ex, label=left:\text{\Large $C_1$}] {}; 
\node[myv4][fit=(v21a) (v3) (v33b) (v21l),  inner xsep=1.5ex, inner ysep=1.5ex, label=below: \text{\Large $C_2$}] {}; 
\node[myv4][fit=(v3) (v4) (v5) (v32b) (v31l) (v42a),  inner xsep=1.5ex, inner ysep=1.5ex, label= below: \text{\Large $C_3$}] {};
\node[myv4][fit= (v5) (v6)(v51b),  inner xsep=1.5ex, inner ysep=1.5ex, label= above: \text{\Large $C_4$}] {};
\node[myv4][fit= (v6) (v62b) (v7) (v71a) (v8) (v9),  inner xsep=1.5ex, inner ysep=1.5ex, label= above: \text{\Large $C_5$}] {};
\node[myv4][fit= (v92b) (v10),  inner xsep=1.5ex, inner ysep=1.5ex, label= above: \text{\Large $C_6$}] {};

\end{tikzpicture}}
    \caption{An interval graph $G$}
    \label{fig:separated_cluster_interval_graph}
  \end{figure}
% Shirish end Example
\begin{figure}
 \centering
  \centering
 %   \vspace{-0.5cm}
    \iffalse
\usetikzlibrary{arrows}
\usetikzlibrary{decorations.markings}
\usetikzlibrary{shapes.geometric}
\usetikzlibrary{positioning}
\usetikzlibrary{decorations.text}
\usetikzlibrary{decorations.pathmorphing}
\fi
\resizebox{0.80\textwidth}{!}{\begin{tikzpicture}[myv1/.style={circle, draw, line width = 1 mm, fill=black},[myv2/.style={circle, draw, line width = 1 mm, fill=none}, myv4/.style={rectangle, draw,dotted,inner sep=0pt, line width = 0.05mm},mvy5/.style={ellipse, draw, line width = 1mm},[mydummy/.style={circle, inner sep=0pt, fill=white}], [myv6/.style = {>={Bracket[width=6mm,line width=1pt,length=1.5mm]},<->}]

\tikzset{
  big arrow/.style={
    decoration={markings,mark=at position 1 with {\arrow[scale=4,#1]{>}}},
    postaction={decorate},
    shorten >=0.4pt},
  big arrow/.default=black}

\node [label =]  (v1L) at (0, 0) {};
\node [label =]  (v1R) at (2, 0) {};

\node [label =] (v2L) at (3, 0) {};
\node [label =] (v2R) at (5, 0) {};

\node [label =] (v3L) at (1, -1.5) {};
\node [label =] (v3R) at (8, -1.5) {};

\node [label =] (v4L) at (7, 0) {};
\node [label =] (v4R) at (9, 0) {};

\node [label =] (v5L) at (6, -4.5) {};
\node [label =] (v5R) at (11, -4.5) {};

\node [label =] (v6L) at (10, -3) {};
\node [label =] (v6R) at (15, -3) {};

\node [label =] (v7L) at (13, -1.5) {};
\node [label =] (v7R) at (17, -1.5) {};

\node [label =] (v8L) at (14, 0) {};
\node [label =] (v8R) at (16, 0) {};

\node [label =] (v9L) at (12, -4.5) {};
\node [label =] (v9R) at (19, -4.5) {};

\node [label =] (v10L) at (18, -3) {};
\node [label =] (v10R) at (21, -3) {};

\draw [>={Bracket[width=6mm,line width=1pt,length=1.5mm]},<->] (v1L) -- (v1R) node[midway, above] {\huge{$I_{v_1}$}};

\draw [>={Bracket[width=6mm,line width=1pt,length=1.5mm]},<->] (v2L) -- (v2R) node[midway, above] {\huge{$I_{v_2}$}};
%\draw [decoration={text along path, raise=10pt,text={\huge{$I_{v_2}$}},text align={center}}, decorate]  (v2L) -- (v2R);

\draw [>={Bracket[width=6mm,line width=1pt,length=1.5mm]},<->] (v3L) -- (v3R) node[midway, above] {\huge{$I_{v_3}$}};

\draw [>={Bracket[width=6mm,line width=1pt,length=1.5mm]},<->] (v4L) -- (v4R) node[midway, above] {\huge{$I_{v_4}$}};

\draw [>={Bracket[width=6mm,line width=1pt,length=1.5mm]},<->] (v5L) -- (v5R) node[midway, above] {\huge{$I_{v_5}$}};

\draw [>={Bracket[width=6mm,line width=1pt,length=1.5mm]},<->] (v6L) -- (v6R) node[midway, above] {\huge{$I_{v_6}$}};

\draw [>={Bracket[width=6mm,line width=1pt,length=1.5mm]},<->] (v7L) -- (v7R) node[midway, above] {\huge{$I_{v_7}$}};

\draw [>={Bracket[width=6mm,line width=1pt,length=1.5mm]},<->] (v8L) -- (v8R) node[midway, above] {\huge{$I_{v_8}$}};

\draw [>={Bracket[width=6mm,line width=1pt,length=1.5mm]},<->] (v9L) -- (v9R) node[midway, above] {\huge{$I_{v_9}$}};

\draw [>={Bracket[width=6mm,line width=1pt,length=1.5mm]},<->] (v10L) -- (v10R) node[midway, above] {\huge{$I_{v_{10}}$}};

\end {tikzpicture}}
    \caption{Interval Representation of $G$ shown in Figure~\ref{fig:separated_cluster_interval_graph}}
    \label{fig:Interval representation of an interval graph1}
  \end{figure}

Note that most of the definitions in this section are based on a fixed interval representation, $\{I_v\}_{v\in V(G)}$ of $G$. Let  $\{C_1,C_2,\ldots, C_x\}$ be the collection of all \textit{maximal cliques} in $G$ (refer Figure~\ref{fig:separated_cluster_interval_graph}) and let $t_i=~ \mid C_i \mid $, for $i\in \{1,2,\ldots,x\}$. Note that,  as $G$ is an interval graph, we have  $x\leq n$~\cite{GolumAlgoGrTh04}.  Now to define the desired collection of cliques, we first introduce the following. We enumerate the vertices belonging to $C_i$ in two different ways (namely, in the increasing order of the left and right end-points of their intervals). i.e. $C_i^{l\_list}=(v^1,v^2,\ldots, v^{t_i})$ is such that $l(v^1)< l(v^2)<\cdots < l(v^{t_i})$, and $C_i^{r\_list}=(u^1,u^2,\ldots, u^{t_i})$ is such that $r(u^1)< r(u^2)<\cdots < r(u^{t_i})$. For an example, refer Table~\ref{tab:lists} that corresponds to the graph in Figure~\ref{fig:separated_cluster_interval_graph} and \ref{fig:Interval representation of an interval graph1}. 

%\vspace{-3.5cm}
\begin{table}[ht]
    \centering
    \tiny{
    \begin{adjustbox}{width=0.9\textwidth}
\begin {tabular} {|l|l|l|}
\hline
  \multicolumn{1}{|c|}{$C_i$} & \multicolumn{1}{c|}{$C_i^{l\_list}$} & \multicolumn{1}{c|}{$C_i^{r\_list}$} \\
\hline
 \makecell[tl]{$C_1 = \{ v_1, v_3\} $\\ $\lvert C_1 \rvert = 2$} &  $C_1^{l\_list} = (v_1, v_3) $ &$C_1^{r\_list} = (v_1, v_3) $ \\
\hline
 \makecell[tl]{$C_2 = \{ v_2, v_3\} $ \\$\lvert C_2 \rvert = 2$} &  $C_2^{l\_list} = (v_3, v_2)$ & $C_2^{r\_list} = (v_2, v_3)$ \\
\hline
 \makecell[tl]{ $C_3 = \{ v_3, v_4, v_5\} $ \\$\lvert C_3 \rvert = 3$}     &   $C_3^{l\_list} = (v_3, v_5, v_4)$ & $C_3^{r\_list} = (v_3, v_4, v_5)$ \\
\hline
 \makecell[tl]{ $C_4 = \{ v_5, v_6\} $ \\$\lvert C_4 \rvert = 2$} & $C_4^{l\_list} = ( v_5, v_6)$ &$C_4^{r\_list} = (v_5, v_6)$ \\
\hline
 \makecell[tl]{ $C_5 = \{ v_6, v_7, v_8, v_9\} $ \\ $\lvert C_5 \rvert = 4$}   &   $C_5^{l\_list} = (v_6, v_9, v_7, v_8)$ & $C_5^{r\_list} = (v_6, v_8, v_7, v_9)$ \\
\hline
 \makecell[tl]{ $C_6 = \{ v_9, v_{10}\} $ \\$\lvert C_6 \rvert = 2$} &  $C_6^{l\_list} = (v_9, v_{10})$&$C_6^{r\_list} = (v_9, v_{10})$ \\
\hline
\end {tabular}

%\bigskip
\end{adjustbox}}
 \caption{The enumerations $C_i^{l\_list}$ and $C_i^{r\_list}$ of the maximal cliques $C_i$ of $G$ shown in Figure~\ref{fig:separated_cluster_interval_graph}}
    \label{tab:lists}
\end{table}
    
    Let $i\in \{1,2,\ldots,x\}$ and $p,q\in \{1,2,\ldots, \mid C_i \mid \}$. 
    We now define $C_i^l(v^p)$ (respectively $C_i^r(u^q)$) as the set of vertices in $C_i$ whose left (respectively, right) end-points of the intervals are greater (respectively less) than or equal to the left (respectively right) end-points interval of  $v^p$ (respectively $u^q$). i.e., $C_i^l(v^p)=\{u\in C_i: l(u) \geq l(v^p)\}$ and $C_i^r(u^q)=\{u\in C_i:  r(u) \leq r(u^q)\}$ (refer Table~\ref{tab:collection}). Clearly, $C_i=C_i^l(v^1)\supseteq C_i^l(v^2)\supseteq \ldots \supseteq C_i^l(v^{t_i})=\{v^1\}$ and $\{u^1\}=C_i^r(u^1)\subseteq C_i^r(u^2)\subseteq \ldots \subseteq C_i^r(u^{t_i})=C_i$. We then define the cliques $C_i^{lr}(v^p,u^q)$ as the set of vertices of $C_i$ whose corresponding intervals are completely contained in the interval $[l(v^p), r(u^q)]$, i.e., $C_i^{lr}(v^p,u^q) = C_i^l(v^p)\cap C_i^r(u^q)$.

\begin{definition}[The collection of cliques $\hat{\mathcal{C}}$]
   For $\{i\in \{1,2,\ldots,x\}\}$,  let $\hat{\mathcal{C}_i}=\{C_i^{lr}(v^p,u^q):p,q\in \{1,2,\ldots, \mid C_i \mid \}\}$.  Then the collection of cliques $\hat{\mathcal{C}}=\bigcup_{i\in[x]}\hat{\mathcal{C}_i}$, for $i\in \{1,2,\ldots,x\}$. 
\end{definition}
%\vspace{-1cm}
\begin{table}[!htb]
   \centering
   {
   \resizebox{0.75\textwidth}{!}{
\begin{adjustbox}{width=\textwidth}
\begin {tabular} {|l|l|l|l|}
\hline
\multicolumn{1}{|c|}{$C_i$} & \multicolumn{1}{c|}{$C_i^l(v^p)$} & \multicolumn{1}{c|}{$C_i^r(u^q)$} & \multicolumn{1}{c|}{$C_i^{lr}(v^p,u^q)$} \\
\hline 
$C_1$ &
\makecell[tl]{
$C_1^{l\_list} = (v_1, v_3)$\\\\
$ C_1^l(v^1) = \{v_1, v_3\} $\\
$C_1^l{(v^2)} = \{v_3\}$
} &
\makecell[tl]{
$C_1^{r\_list} = (v_1, v_3)$\\\\
$ C_1^r(v^1) = \{v_1\} $\\
$C_1^r(v^2) = \{v_1, v_3\}$
} &
\makecell[tl]{\\
$C_1^{lr}(v^1,v^1) = \{v_1\} $ \\
$C_1^{lr}(v^1,v^2)=\{v_1, v_3\}$\\
%$C_1^{lr}(v_l^2,v_r^1) = \phi $\\
$C_1^{lr}(v^2,v^2) = \{v_3\}$ 
} \\
\hline

$C_2$ & 
\makecell[tl]{
$C_2^{l\_list} = (v_3, v_2)$\\\\
$ C_2^l(v^1) = \{v_3, v_2\}$\\
$C_2^l(v^2) = \{v_2\}$
} &
\makecell[tl]{
$C_2^{r\_list} = (v_2, v_3)$\\\\
$C_2^r(v^1) = \{v_2\}$\\
$C_2^r(v^2) = \{v_2, v_3\}$
} &
\makecell[tl]{\\
$C_2^{lr}(v^1,v^1) = \{v_2\}$ \\
$C_2^{lr}(v^1,v^2) = \{v_2, v_3\}$\\
$C_2^{lr}(v^2,v^1) = \{v_2\}$\\
$C_2^{lr}(v^2,v^2) = \{v_2\}$ 
} \\
\hline

$C_3$ & 
\makecell[tl]{
$C_3^{l\_list} = (v_3, v_5, v_4)$\\\\
$C_3^l(v^1) = \{v_3, v_5, v_4\} $\\
$C_3^l(v^2) = \{v_5, v_4\}$\\
$C_3^l(v^3) = \{v_4\}$
} &
\makecell[tl]{
$C_3^{r\_list} = (v_3, v_4, v_5)$\\\\
$C_3^r(v^1) = \{v_3\}$\\
$C_3^r(v^2) = \{v_3, v_4\}$\\
$C_3^r(v^3) = \{v_3, v_4, v_5\}$
} &
\makecell[tl]{\\
$C_3^{lr}(v^1,v^1) = \{v_3\}$ \\
$C_3^{lr}(v^1,v^2) = \{v_3, v_4\}$\\
$C_3^{lr}(v^1,v^3) = \{v_3, v_4, v_5\}$\\
%$C_3^{lr}(v_l^2,v_r^1) = \phi$\\
$C_3^{lr}(v^2,v^2) = \{v_4\}$ \\
$C_3^{lr}(v^2,v^3) = \{v_4, v_5\}$ \\
%$C_3^{lr}(v_l^3,v_r^1) = \phi$\\
$C_3^{lr}(v^3,v^2) = \{v_4\}$ \\
$C_3^{lr}(v^3,v^3) =\{v_4\}$ 
} \\
\hline
$C_4$ &
\makecell[tl]{
$C_4^{l\_list} = (v_5, v_6)$\\\\
$C_4^l(v^1) = \{v_5, v_6\} $\\
$C_4^l(v^2) = \{v_6\}$
} &
\makecell[tl]{
$C_4^{r\_list} = (v_5, v_6)$\\\\
$ C_4^r(v_r^1) = \{v_5\} $\\
$C_4^r(v_r^2) = \{v_5, v_6\}$
} &
\makecell[tl]{\\
$C_4^{lr}(v^1,v^1) = \{v_5\} $ \\
$C_4^{lr}(v^1,v^2)= \{v_5, v_6\}$\\
%$C_46{lr}(v_l^2,v_r^1) = \phi$\\
$C_4^{lr}(v^2,v^2) = \{v_6\}$ 
} \\
\hline
$C_5$ & 
\makecell[tl]{
$C_5^{l\_list} = (v_6, v_9, v_7, v_8)$\\\\
$C_5^l(v^1) = \{v_6, v_9, v_7, v_8\} $\\
$C_5^l(v^2) = \{v_9, v_7, v_8\}$\\
$C_5^l(v^3) = \{v_7, v_8\}$\\
$C_5^l(v^4) = \{v_8\}$
} &
\makecell[tl]{
$C_5^{r\_list} = (v_6, v_8, v_7, v_9)$\\\\
$ C_5^r(v^1) = \{v_6\} $\\
$C_5^r(v^2) = \{v_6, v_8\}$\\
$C_5^r(v^3) = \{v_6, v_8, v_7\}$\\
$C_5^r(v^4) = \{v_6, v_8, v_7, v_9\}$
} &
\makecell[tl]{\\
$C_5^{lr}(v^1,v^1) = \{v_6\}$ \\
$C_5^{lr}(v^1,v^2) = \{v_6, v_8\}$\\
$C_5^{lr}(v^1,v^3) = \{v_6,v_7, v_8\}$\\
$C_5^{lr}(v^1,v^4) = \{v_6, v_7, v_8, v_9\}$\\
%$C_5^{lr}(v_{l2}}{v_{r1}}} = \phi$\\
$C_5^{lr}(v^2,v^2) = \{v_8\}$ \\
$C_5^{lr}(v^2,v^3) = \{v_7, v_8\}$ \\
$C_5^{lr}(v^2,v^4) = \{v_7, v_8, v_9\}$ \\
%$C_5^{lr}(v_l3}}{v_{r1}}} = \phi$\\
$C_5^{lr}(v^3,v^2) = \{v_8\}$ \\
$C_5^{lr}(v^3,v^3) = \{v_7, v_8\}$ \\
$C_5^{lr}(v^3,v^4) = \{v_7, v_8\}$ \\
%$C_5^{lr}(v_l^4,v_r^1) = \phi$\\
$C_5^{lr}(v^4,v^2) = \{v_8\}$ \\
$C_5^{lr}(v^4,v^3) = \{v_8\}$ \\
$C_5^{lr}(v^4,v^4) = \{v_8\}$ 
} \\
\hline
$C_6$ &
\makecell[tl]{
$C_6^{l\_list} = (v_9, v_{10})$\\\\
$C_6^l(v^1) = \{v_9, v_{10}\} $\\
$C_6^l(v^2) = \{v_{10}\}$
} &
\makecell[tl]{
$C_6^{r\_list} = (v_9, v_{10})$\\\\
$C_6^r(v^1) = \{v_9\} $\\
$C_6^r(v^2) = \{v_9, v_{10}\}$
} &
\makecell[tl]{\\
$C_6^{lr}(v^1,v^1) = \{v_9\}$ \\
$C_6^{lr}(v^1,v^2)= \{v_9, v_{10}\}$\\
%$C_6^{lr}(v_l^2,v_r^1) = \phi$\\
$C_6^{lr}(v^2,v^2) = \{v_{10}\}$ 
} \\
\hline

\end {tabular}
\end{adjustbox}}
%\medskip

 \caption{The collection of cliques $C_i^l$, $C_i^r$, and $C_i^{lr}$ corresponding to the maximal cliques $C_i$ of the graph $G$ shown in Figure~\ref{fig:separated_cluster_interval_graph}}
    \label{tab:collection}
}
\end{table} 

To illustrate the above definitions, consider the interval graph $G$ in Figure~\ref{fig:separated_cluster_interval_graph} and its interval representation given in Figure~\ref{fig:Interval representation of an interval graph1}. It is easy to see that the collection of maximal cliques in $G$ is $\{C_1, C_2, \ldots, C_6\}$. For each $C_i$, the enumerations of $C_i$, namely, $C_i^{l\_list}$ and $C_i^{r\_list}$ are listed in Table~\ref{tab:lists}. Further, for each $C_i$, where $i\in \{1,2,\ldots,6\}$, in Table~\ref{tab:collection}, we summarize the collection of cliques $C_i^l(v_p)$, $C_i^r(u^q)$, and $C_i^{lr}(v^p,u^q)$. Note that, in the third column of Table~\ref{tab:collection}, we did not include entries corresponding to those pairs $p,q\in \{1,2,\ldots,\lvert C_i \rvert \}$ for which $C_i^{lr}(v^p,u^q)=\emptyset$. Now, for each $i\in \{1,2,\ldots,6\}$, we have $\hat{\mathcal{C}_i}=\{C_i^{lr}(v^p,u^q):p,q\in \{1,2,\ldots, \mid C_i \mid \}\}$; i.e., $\hat{\mathcal{C}}=\bigcup_{i\in[x]} \hat{\mathcal{C}_i}:i\in \{1,2,\ldots,6\}$.\\

We note the following.

\begin{remark} \label{rem:clique_complexity}
Since $C_i \in \hat{\mathcal{C}_i}$ for each $i \in \{1, 2, \ldots, x \}$ (as $C_i = C_i^{lr}(v^1, u^{\lvert C_i \rvert}))$, the collection of cliques  $\hat{\mathcal{C}}$ contains all the maximal cliques in $G$. Also, as the number of ordered pairs of the form $(p,q)$ is at most $n^2$ (as $p,q\in \{1,2,\ldots, \mid C_i \mid \}$, $ \mid C_i \mid \leq n$), we then have that $ \mid \hat{\mathcal{C}} \mid \leq n^3$. 

\end{remark}

\noindent\textbf{Properties of separated-cluster in interval graphs:}
Let $\mathcal{F}$ be a family of disjoint cliques in an interval graph $G$. With respect to an interval representation $\{I_v\}_{v\in V(G)}$ of $G$, we can define a natural ordering $(S_1,S_2,\ldots,S_k)$ of the cliques in $\mathcal{F}$ as follows: \textit{ for any pair $i,j\in \{1,2,\ldots,k\}$, we define $S_i<S_j$ if and only if the Helly region  (the common intersection region of intervals belonging to the same clique) of the clique $S_i$, denoted as $H_{S_i}$, lies entirely to the left of the Helly region of the clique $S_j$}. We will assume this ordering throughout this section. For a clique $S_i\in \mathcal{S}$, where $i\in \{1,2,\ldots,k\}$, let $C_i$ be a maximal clique in $G$ such that $S_i\subseteq C_i$ (possibly, $C_i=S_i$). 

Let $\mathcal{S}=\{S_1,S_2,\ldots,S_k\}$ be a separated-cluster in $G$. Let $C_i$ be a maximal clique in $G$ such that $S_i\subseteq C_i$. Let $j\in \{1,2,\ldots,k\}$, with $j\neq i$. 
We then have the following observations.

\begin{observation}\label{obs:propmaxclique}

(i) $S_j\cap C_i=\emptyset$\\ (ii) no vertex in $C_i$ is adjacent to any vertex in $S_j$.
\end{observation}
%\vspace{-0.75cm}
\begin{proof}
    To prove (i), for the sake of contradiction assume that  $S_j\cap C_i\neq \emptyset$, and let $z\in S_j\cap C_i$. %Since $S_i\cap S_j=\emptyset$ (by the definition of separated-cluster), we have $z\in S_j$ and $z\in C_i\setminus S_i$. 
    Then $z$ is adjacent to all the vertices in $S_i$, as $C_i$ is a clique and $S_i\subseteq C_i$, which contradicts the fact that $\mathcal{S}$ is a separated cluster. 
     (ii) Suppose that there exists a vertex, say $z\in C_i$, such that $z$ is adjacent to some vertex in $S_j$. Clearly, by the definition of $\mathcal{S}$, we have $z\notin S_i$ and by (i), $S_j\cap C_i=\emptyset$. %Thus $Z\notin S_i$. If either $S_i$ and $S_j$ have adjacent vertices, or there exists a vertex $z\in C_i\setminus S_i$ such that 
    This implies that $z$ is a common neighbor of some vertex in $S_i$ and some vertex in $S_j$, where $i\neq j$. This again contradicts the definition of the separated cluster $\mathcal{S}$.
\end{proof}

A vertex $x\in C_i\setminus S_i$ is said to have a \textit{left-conflict} with respect to $\mathcal{S}$, 
if there exists $S_j\in \mathcal{S}$, with $j<i$, and $\exists z\in V(G)\setminus(S_i\cup S_j)$ such that $z$ is a common neighbor of both $x$ and some vertex $y\in S_j$, i.e., $z\in N(x)\cap N(y)$. Similarly, a vertex $x\in C_i\setminus S_i$ is said to have a \textit{right-conflict} with respect to $\mathcal{S}$, 
if there exists $S_j \in \mathcal{S}$ with $j>i$ and $\exists z\in V(G)\setminus(S_i\cup S_j)$ such that $z$ is a common neighbor of both $x$ and some vertex $y\in S_j$, i.e., $z\in N(x)\cap N(y)$. 

 The following observation is due to the definitions of \textit{left-conflict} (respectively \textit{right-conflict}) and $C_i^l(v^p)$ (respectively $C_i^r(u^q)$).
\begin{observation}\label{obs:conflict}
     Let $v$ be a vertex in $C_i\setminus S_i$ which has a \textit{left-conflict}  with respect to $\mathcal{S}$, where $v=v^p$, $p\in \{1,2,\ldots, \mid C_i \mid \}$ with respect to the ordering $C_i^{l\_list}$ of $C_i$. Then every vertex in $C_i\setminus C_i^l(v^p)$  has a \textit{left-conflict} (i.e. $S_i\subseteq C_i^{l}(v^{p+1})$). Similarly, let $v\in C_i\setminus S_i$ be a vertex which has a  \textit{right-conflict} with respect to $\mathcal{S}$, where $v=u^q$, $q\in \{1,2,\ldots, \mid C_i \mid \}$ with respect to the ordering  $C_i^{r\_list}$ of $C_i$. Then every vertex in  $C_i\setminus C_i^r(u^q)$ has a \textit{right-conflict} (i.e. $S_i\subseteq C_i^{r}(u^{q-1})$).
     
\end{observation}

We also note the following.
\begin{observation} \label{obs:noconflict}
%Let $\mathcal{S}=\{S_1,S_2,\ldots,S_k\}$ be a separated-cluster in $G$. For $i\in \{1,2,\ldots,k\}$, let $C_i$ be a maximal clique in $G$ such that $S_i\subseteq C_i$.
If there exists a vertex $x\in C_i\setminus S_i$ such that $x$ has neither a left-conflict nor a right-conflict with respect to $\mathcal{S}$, then $\mathcal{S'}=\{S_1,S_2,\ldots, S_i',\ldots,S_k\}$ is also a separated-cluster in $G$, where $S_i'=S_i\cup\{x\}$.
\end{observation}
\vspace{-0.5cm}
\begin{proof}
    To prove this, it is enough to show that $\mathcal{S'}$ satisfies the conditions in the definition of separated-clusters. Let $j\in \{1,2,\ldots,k\}$ with $j\neq i$. Since $\mathcal{S}$ is a separated-cluster in $G$, by Observation~\ref{obs:propmaxclique}, we have that $C_i\cap S_j=\emptyset$ and no vertex in $C_i$ can be adjacent to any vertex in $S_j$. This implies that $S_i'\cap S_j=\emptyset$  and $S_i'$ is non-adjacent to $S_j$. Suppose that there exists a vertex $y\in S_j$ and $z\in V(G)\setminus (S_i\cup S_j)$ such that $x,y\in N_G(z)$. Since $x\in C_i$, we then have that $x$ has a \textit{left-conflict} with respect to $\mathcal{S}$ if $S_j<S_i$, and $x$ has a \textit{right-conflict} with respect to $\mathcal{S}$ if $S_i<S_j$. In either case, we have a contradiction to the choice of $x$. Now since $\mathcal{S}$ is a separated-cluster, we can therefore conclude that $\mathcal{S'}=\{S_1,S_2,\ldots, S_i',\ldots,S_k\}$ is also a separated-cluster in $G$. 
\end{proof}

In the following lemma, we observe a crucial property of a maximum cardinality separated-cluster.

\begin{lemma} \label{lem:maxsep}
    Let $\mathcal{S}=\{S_1,S_2,\ldots,S_k\}$ be a maximum cardinality
    separated-cluster 
    in $G$. Then for each $S_i$, for $i\in \{1,2,\ldots,k\}$, we have %$C_i\setminus S_i$ has either left-conflict or right-conflict. 
    $S_i\in \hat{\mathcal{C}}$.
\end{lemma}
\vspace{-0.5cm}

\begin{proof}

    For $i\in \{1,2,\ldots,k\}$, let $S_i\in \mathcal{S}$. If $S_i$ is a maximal clique in $G$, then $S_i\in \hat{\mathcal{C}}$ (since $\hat{\mathcal{C}}$ contains all the maximal cliques in $G$). Suppose that $S_i$ is not a maximal clique in $G$. Then, there exists a maximal clique, say $C_i$ in $G$, such that $S_i\subseteq C_i$. Clearly, $C_i\setminus S_i \neq \emptyset$. Since $\mathcal{S}$ is a maximum cardinality separated-cluster in $G$, by Observation~\ref{obs:noconflict}, we have that every vertex in $C_i\setminus S_i$ has either a \textit{left-conflict} or a \textit{right-conflict} or both. 

   If $C_i\setminus S_i$ has at least one vertex with a left-conflict (respectively,  right-conflict), then let $v^p\in C_i\setminus S_i$ (respectively, $u^q \in C_i\setminus S_i$) be the vertex having maximum left-end point (respectively, minimum right-end point) with respect to $C_i^{l\_list}$ (respectively, $C_i^{r\_list}$)  such that $v^p$ has a left-conflict (respectively, $u^q$ has a right-conflict). Then, by Observation~\ref{obs:conflict}, we have $S_i \subseteq C_i^l(v^{p+1})$ (respectively, $S_i \subseteq C_i^r(u^{q-1})$). 
   %and $u^q \in C_i\setminus S_i$ be the vertex having the minimum right-end point with respect to $C_i^{r\_list}$ such that $u^q$ (if exists) has a right-conflict. 
   And, if  $C_i\setminus S_i$ does not have a vertex having a left-conflict (respectively, right-conflict), then we consider $v^{p+1}$ as $v^1$ (respectively, $u^{q-1}$ as $u^{ \mid C_i \mid }$).  %any case, by Observation~\ref{obs:conflict}, we have $S_i \subseteq C_i^l(v^{p+1})\cap C_i^r(u^{q-1}) = C_i^{lr}(v^{p+1},u^{q-1})$. 
   Therefore, in any case, we have, $S_i\subseteq  C_i^l(v^{p+1})\cap C_i^r(u^{q-1})=C_i^{lr}(v^{p+1},u^{q-1})$, for some $p,q\in \{0,1,\ldots, \mid C_i \mid +1\}$.

    Now to prove the reverse inequality, first note that by the definitions of $v^p$, $u^q$, and $C_i^{lr}$, no vertex in $C_i^{lr}(v^{p+1},u^{q-1})\subseteq C_i$ can have a \textit{left-conflict} or a \textit{right-conflict} with respect to  $\mathcal{S}$. Therefore, by Observation~\ref{obs:noconflict}, and the fact that $\mathcal{S}$ is a maximum cardinality separated-cluster, we can conclude that $C_i^{lr}(v^{p+1},u^{q-1}) \subseteq S_i$. Thus,  $S_i=C_i^{lr}(v^{p+1},u^{q-1})$ for some $p,q\in \{0,1,\ldots, \mid C_i \mid +1\}$.
Since $S_i = C_i^{lr}(v^{p+1},u^{q-1})\in \hat{\mathcal{C}}$, the statement of the lemma follows.
\end{proof}

We use the following construction to reduce the \SCP\ problem on interval graphs to the maximum weighted independent set problem on cocomparability graphs.

\begin{construction}[Weighted conflict graph $G_c$] \label{def:conflict}
    Given an interval graph $G$ having maximal cliques, say $\{C_1,C_2,\ldots,C_l\}$, the weighted conflict graph $G_c$ is defined as follows: $V(G_c)=\hat{\mathcal{C}}=\{C_i^{lr}(v^p, u^q): i\in \{1,2,\ldots,l\}, p,q\in \{1,2,\ldots, \mid C_i \mid \}\}$ (each vertex in $G_c$ represents a clique in $G$, which is also a member of $\hat{\mathcal{C}}$).  For the sake of convenience, let $V(G_c)=\hat{\mathcal{C}}=\{K_1,K_2,\ldots,K_t\}$, where $t= \mid \hat{\mathcal{C}} \mid \leq n^3$. Now, for each vertex, say $K_j\in V(G_c)$, we define the weight function $w(K_j)= \mid K_j \mid $ (i.e., the cardinality of the corresponding clique $K_j$ in $G$). For any pair of vertices, say $K_i,K_j\in V(G_c)$ with $i\neq j$, we make $K_i$ and $K_j$ adjacent in $G_c$ if and only if at least one of the following conditions is true in $G$.
    \begin{enumerate}[label=(\alph*)]
        \item \label{cond:1} $K_i\cap K_j\neq \emptyset$, or
        \item \label{cond:2} there exist vertices, $x\in K_i$ and $y\in K_j$ such that $xy\in E(G)$, or
        \item \label{cond:3} there exist vertices, $x\in K_i$, $y\in K_j$, and $z\in V(G)\setminus (K_i\cup K_j)$ such that $x,y\in N_G(z)$.
    \end{enumerate}
\end{construction}

We claim that  $G_c$ is a cocomparability graph. To prove this,  we use the following vertex ordering characterization of cocomparability graphs  due to Damaschke~\cite{DamForbidOS90}.

\begin{theorem}[~\cite{DamForbidOS90}]\label{thm:cocomp}
    An undirected graph $H$ is a cocomparability graph if and only if there is an ordering $<$ of $V(H)$ such that for any three vertices $i<j<k$, if $ik\in E(H)$, then either $ij\in E(H)$ or $jk\in E(H)$.
\end{theorem}
The vertex ordering specified in Theorem~\ref{thm:cocomp} is called \textit{umbrella-free} ordering~\cite{DamForbidOS90} which essentially says that \textit{cocomparability graphs are exactly those graphs whose vertex set admits an umbrella-free} ordering.

Let $G$ be an interval graph with an interval representation, $\{I_v\}_{v\in V(G)}$, where $n= \mid V(G) \mid $ and $G_c$ be the corresponding weighted conflict graph obtained by Construction~\ref{def:conflict}.  Recall that $V(G_c)=\{K_1,K_2,\ldots,K_t\}$, where $t=~\mid\hat{\mathcal{C}} \mid~ \leq ~n^3$. Let $j\in \{1,2,\ldots,t\}$. Since each of the vertices $K_j$ in $V(G_c)$ represent a clique in $G$, they have a \textit{Helly region} (the common intersection region of intervals belonging to the same clique), say $H_{K_j}$ in the interval representation, $\{I_v\}_{v\in V(G)}$.  Let $r(H_{K_j})$ denote the right end-point of the \textit{Helly region} to the clique corresponding to $K_j$.  We then define an ordering $<$ of the vertices in $G_c$ as follows: \textit{for any two vertices $K_i,K_j\in V(G_c)$, where $i,j\in \{1,2,\ldots,t\}$, we say that $K_i<K_j$ if and only if $r(H_{K_i})\leq r(H_{K_j})$} (break the ties arbitrarily, i.e. for any pair $i,j\in \{1,2,\ldots,t\}$,  if $r(H_{K_i})= r(H_{K_j})$ then we can have either $K_i<K_j$ or $K_j<K_i$). \\

We then have the following lemma. 

\begin{lemma}
    For any interval graph $G$, the weighted conflict graph $G_c$ is a cocomparability graph.
\end{lemma}
\begin{proof}
    To prove the lemma, it is enough to show that the ordering $<$ of $V(G_c)$ (defined in the paragraph above) is an umbrella-free ordering. Suppose not. Then there exist cliques, say, $K_i,K_j,K_l$ in $V(G_c)$ such that $K_i<K_j<K_l$, $K_iK_l\in E(G_c)$, but $K_iK_j\notin E(G_c)$ and $K_jK_l\notin E(G_c)$. Since $K_i<K_j<K_l$, we have by the definition of $<$ that $r(H_{K_i})\leq r(H_{K_j})\leq r(H_{K_l})$ (where $H_{K_i}$, $H_{K_j}$, and $H_{K_l}$ denote the \textit{Helly regions} of the cliques, $K_i$, $K_j$, and $K_l$ respectively). Note that $K_iK_l\in E(G_c)$. Therefore, by the definition of $G_c$, at least one of the conditions in Construction~\ref{def:conflict} has to be true. 

    \noindent\textbf{Case-1:} \textit{The edge $K_iK_l\in E(G_c)$ is due to Condition~\ref{cond:1}}. 

    \noindent This implies that $K_i\cap K_l\neq \emptyset$. Let $x\in K_i\cap K_l$. Since $r(H_i)\leq r(H_j)\leq r(H_l)$, we then have that the interval representing the vertex $x$ intersects with the \textit{Helly region}, $H_j$ of the clique $K_j$. Therefore, by Condition~\ref{cond:1} or~\ref{cond:2} in Construction~\ref{def:conflict} (depending on the fact that the vertex $x$ belongs to $K_j$ or not), we have that both the edges, $K_iK_j, K_jK_l\in E(G_c)$. This is a contradiction.

    \noindent\textbf{Case-2:} \textit{The edge $K_iK_l\in E(G_c)$ is due to Condition~\ref{cond:2}}.

    \noindent This implies that there exist vertices $x\in K_i$ and $y\in K_l$ such that $xy\in E(G)$. Note that for each vertex $u\in K_i$, we have $r(u)< l(H_j)\leq r(H_j)$ (Since by assumption, $K_iK_j\notin E(G_c)$ which will be violated due to  Conditions~\ref{cond:1} and~\ref{cond:2} of Construction~\ref{def:conflict}). %(since $r(H_i)\leq r(H_j)$, $K_iK_j\notin E(G_c)$, and by Conditions~\ref{cond:1} and~\ref{cond:2} of Construction~\ref{def:conflict}). 
    Similarly, we have $r(H_j)<l(v)$ for each vertex $v\in K_l$ (as $r(H_j)\leq r(H_l)$, $K_jK_l\notin E(G_c)$, and by Conditions~\ref{cond:1} and~\ref{cond:2} of Construction~\ref{def:conflict}). Therefore, as $x\in K_i$ and $y\in K_l$, in particular, we have $r(x)<r(H_j)<l(y)$. This implies that $I_x\cap I_y=\emptyset$. Therefore, $xy\notin E(G)$, a contradiction.

    \noindent\textbf{Case-3:} \textit{The edge $K_iK_l\in E(G_c)$ is due to Condition~\ref{cond:3}}.
    
   \noindent This implies that there exist vertices $x\in K_i$, $y\in K_l$, and $z\in V(G)\setminus (K_i\cup K_l)$ such that $x,y\in N_G(z)$. Note that $xy\notin E(G)$, as we already have a contradiction for this in Case-2. Now since $r(H_i)\leq r(H_j)\leq r(H_l)$ and $xz,yz\in E(G)$, %$K_iK_j\notin E(G_c)$, and $K_jK_l\notin E(G_c)$, 
   this further implies that the $l(z)\leq r(x) < l(y)\leq r(z)$. Since the interval of $x$ starts before $r(H_i)$ and interval of $y$ ends after $r(H_l)$, the interval from $r(H_i)$ to $r(H_l)$ is covered by the intervals of $x,y$ and $z$.
 
 In either cases, by Condition~\ref{cond:1} or~\ref{cond:2} or~\ref{cond:3} in Construction~\ref{def:conflict} we have at least one of the edges, $K_iK_j, K_jK_l\in E(G_c)$. This is again a contradiction.

   Since we obtain a contradiction in all the possible cases, we can therefore conclude that the ordering $<$ of $V(G_c)$ (defined in the paragraph above) is an umbrella-free ordering. This implies that $G_c$ is a cocomparability graph by Theorem~\ref{thm:cocomp}. Hence the lemma.
\end{proof}

\noindent\textbf{Reduction:} Here, we show a polynomial-time reduction of \SCP\ problem on interval graphs to the maximum weighted independent set problem on cocomparability graphs. We first note the following theorem due to K\"ohler and Mouatadid~\cite{KohLalMaxWtIndSetAlg16}.
\begin{theorem}[\cite{KohLalMaxWtIndSetAlg16}] \label{thm:indptcocomp}
    Let $H$ be a cocomparability graph with weight function\\ $w:~V(G)~\rightarrow~R^+$. Then an independent set of maximum possible weight in $H$ can be computed in $O( \mid V(H)+ \mid E(H) \mid )$ time.
\end{theorem}

%\vspace{-0.5cm}

We then have the following main theorem.

\begin{theorem}\label{thm:reduction}
    Let $G$ be an interval graph and $G_c$ be its weighted conflict graph. Then $\mathcal{S}=\{S_1,S_2,\ldots,S_k\}$ is a maximum cardinality separated-cluster in $G$ if and only if $\mathcal{S}=\{S_1,S_2,\ldots,S_k\}$ is a maximum weighted independent set in $G_c$.
\end{theorem}
\vspace{-0.5cm}
\begin{proof}
First, note that if $\mathcal{S}=\{S_1,S_2,\ldots,S_k\}$ is a maximum cardinality separated-cluster in $G$, then by Lemma~\ref{lem:maxsep}, we have $\mathcal{S}=\{S_1,S_2,\ldots,S_k\}\subseteq \hat{\mathcal{C}}=V(G_c)$. Moreover by the definition of separated-clusters, the conditions for $S_i,S_j\in \mathcal{S}$ is consistent with the conditions~\ref{cond:1}, \ref{cond:2}, and~\ref{cond:3} of the weighted conflict graph $G_c$. Therefore, $\mathcal{S}\subseteq V(G_c)$ is an independent set in $G_c$.
On the other hand, if  $\mathcal{S}=\{S_1,S_2,\ldots,S_k\}$ is a independent set in $G_c$ then clearly,  for $S_i,S_j\in \mathcal{S}$ the conditions~\ref{cond:1}, \ref{cond:2} and~\ref{cond:3}  are satisfied. Therefore, $\mathcal{S}$ is a separated-cluster in $G$. Since weights of the vertices in $G_c$ (i.e. cliques in $G$) are exactly the cardinality of their corresponding sets, we can therefore conclude that $\mathcal{S}\subseteq V(G_c)=\hat{\mathcal{C}}$ is a maximum weighted independent set in $G_c$ if and only if $\mathcal{S}$ is a maximum cardinality separated-cluster in $G$. This proves the theorem.
\end{proof}

Now Theorem~\ref{thm:sepiterval} follows from Theorems~\ref{thm:indptcocomp} and~\ref{thm:reduction}.
\section{Conclusion}
In this paper, we contribute to the literature connecting the parameters, $\XCD(G)$, \raisebox{2pt}{$\gamma$}$_t(G)$, and $\SC(G)$ of a graph $G$.  Here, we introduce the notion of $cd$-perfectness, and also provide a sufficient condition for a graph to be $cd$-perfect. Even though we also have some necessary conditions for the same, the problem of \textit{characterizing $cd$-perfect graphs} remains open for general graphs. Further, we also prove that \SCP\ is polynomial-time solvable for interval graphs. But, as interval graphs are not necessarily $cd$-perfect, \textit{the algorithmic complexity of \CDC\ on interval graphs} is still open.
\bibliographystyle{elsarticle-num}

\RaggedRight%\setlength{\emergencystretch}{.5em}
\bibliography{main}

\end{document}